\documentclass[journal]{IEEEtran}
\IEEEoverridecommandlockouts
\usepackage{cite}
\usepackage{amsmath,amssymb,amsfonts}
\usepackage{algorithmic}
\usepackage{graphicx}
\usepackage{textcomp}
\usepackage{xcolor}
\def\BibTeX{{\rm B\kern-.05em{\sc i\kern-.025em b}\kern-.08em
    T\kern-.1667em\lower.7ex\hbox{E}\kern-.125emX}}

\usepackage{epsfig,graphics,subfigure}
\usepackage{amsmath,amssymb,amsthm,multirow,balance,dsfont,hyperref}
\usepackage{cite,cleveref}
\usepackage{colortbl}
\usepackage{array}
\usepackage[ruled,vlined]{algorithm2e}
\usepackage{mathrsfs}
\usepackage{empheq}
\usepackage[mathscr]{euscript}
\usepackage{rotating}
\usepackage{caption}
\usepackage{hyperref}
\usepackage{cleveref}

\renewcommand{\thetable}{\arabic{table}}
\crefname{table}{Table}{Tables}
\DeclareFontFamily{U}{mathx}{\hyphenchar\font45}
\DeclareFontShape{U}{mathx}{m}{n}{
      <5> <6> <7> <8> <9> <10>
      <10.95> <12> <14.4> <17.28> <20.74> <24.88>
      mathx10
      }{}
\DeclareSymbolFont{mathx}{U}{mathx}{m}{n}
\DeclareFontSubstitution{U}{mathx}{m}{n}
\DeclareMathAccent{\widecheck}{0}{mathx}{"71}

\newtheorem{theorem}{Theorem}
\newtheorem{lemma}{Lemma}

\newtheorem{assumption}{Assumption}

\usepackage{color}

\newcommand{\boldh}{\boldsymbol{h}}

\newcommand{\bd}{\boldsymbol{d}}

\newcommand{\bs}{\boldsymbol{s}}

\newcommand{\bu}{\boldsymbol{u}}

\newcommand{\bw}{\boldsymbol{w}}
\newcommand{\bz}{\boldsymbol{z}}

\newcommand{\bx}{\boldsymbol{x}}

\newcommand{\bgamma}{\boldsymbol{\gamma}}
\newcommand{\bpsi}{\boldsymbol{\psi}}

\newcommand{\bH}{\boldsymbol{H}}

\newcommand{\bJ}{\boldsymbol{J}}

\newcommand{\bN}{\boldsymbol{N}}

\newcommand{\cA}{\mathcal{A}}

\newcommand{\cJ}{\mathcal{J}}

\newcommand{\cN}{\mathcal{N}}

\newcommand{\cV}{\mathcal{V}}

\newcommand{\cw}{{\scriptstyle\mathcal{W}}}

\newcommand{\bcB}{\boldsymbol{\cal{B}}}
\newcommand{\bcH}{\boldsymbol{\cal{H}}}

\newcommand{\bcD}{\boldsymbol{\cal{D}}}
\newcommand{\bcF}{\boldsymbol{\cal{F}}}
\newcommand{\bcw}{\boldsymbol{\cw}}

\newcommand{\bwt}{\widetilde \bw}

\newcommand{\bcwt}{\widetilde\bcw}

\newcommand{\bcwc}{\widecheck\bcw}

\newcommand{\cwb}{\overline{\cw}}

\newcommand{\expec}{\mathbb{E}}

\newcommand{\col}{\text{col}}

\newcommand{\diag}{\text{diag}}
\newcommand{\sign}{\text{sign}}

\usepackage{etoolbox}
\AtBeginEnvironment{tabular}{\scriptsize}
\newcolumntype{C}[1]{>{\centering\arraybackslash}m{#1}}

\begin{document}
\title{Decentralized collaborative continual learning: \\A multi-objective minimization-based technique  {\footnotesize 
}

\thanks{This work was supported by the ANR JCJC grant ANR-22-CE23-0015-01 (CEDRO project).

A short conference version of this work appears in~\cite{zgheib2025decentralized}. This extended version includes proofs, derivations, and new results.

 Y. Zgheib, M. Antonini, and R. Nassif are with Universit\'e C\^ote d'Azur, I3S Laboratory, CNRS,  France (emails: $\{$yara.zgheib, marc.antonini, roula.nassif$\}$@unice.fr).}}

\author{Yara Zgheib, \IEEEmembership{Student Member, IEEE}, Marc Antonini, \IEEEmembership{Member, IEEE}, Roula Nassif, \IEEEmembership{Member, IEEE}  \\

}

\maketitle

\begin{abstract}
In this work, we formulate decentralized continual learning within a multi-objective optimization framework. For a given inference task~$t$ (corresponding to a common minimizer shared by the cost functions of all agents), agents collecting data in a distributed and streamed manner are only allowed to perform local computations and to exchange information with neighboring agents over the underlying communication graph. As tasks evolve sequentially over time, agents must adapt to newly arriving tasks while retaining knowledge acquired from previously learned ones. This requirement leads to the well-known stability--plasticity dilemma, where stability refers to the ability to retain previous knowledge, while plasticity refers to the ability to learn and adapt to new tasks. To address the stability challenge, agents store subsets of samples from past tasks in local memory buffers. Then, through an appropriate multi-objective formulation, the stored information is incorporated into the learning process so that parameter updates account jointly for the current task and previously learned tasks.  The proposed decentralized continual learning approach is analyzed  in the mean-square-error sense under general assumptions on the individual cost functions and gradient noise processes. The analysis reveals that cooperation among agents improves the performance of continual learning. In particular, by exchanging information with neighboring agents, decentralized collaborative learning can exploit the diversity of locally observed data and memory buffers to improve the network average mean-square deviation (MSD) across tasks. Finally, simulations illustrate the theoretical findings and the effectiveness of the method in reducing forgetting and improving the average MSD across tasks.
\end{abstract}
\vspace{-0.3cm}
\begin{IEEEkeywords}
Decentralized learning, diffusion adaptation, continual learning, gradient noise, mean-square-error stability.
\end{IEEEkeywords}
\vspace{-0.6cm}
\section{Introduction}
\vspace{-0.1cm}
The increasing amount of distributed data generated by applications such as smartphones, drones, autonomous vehicles, and IoT systems, arriving continuously with evolving distributions over time, has created a need for scalable learning algorithms capable of adapting to dynamic environments. Although this data presents opportunities for advancing intelligent systems, it also raises major challenges. Learning models must preserve privacy, operate under limited computational and communication ressources, and continuously adapt to changing data distributions without costly re-learning from scratch or forgetting previously acquired knowledge. These requirements motivate the combination of  \emph{continual learning} and \emph{decentralized learning}, two research directions that have traditionally been studied independently. 

In \emph{decentralized learning} (also known as \emph{distributed learning}), multiple agents (referred to as nodes, clients, or devices) collaborate to learn either a shared global model or multiple related personalized models without relying on a central server~\cite{sayed2014adaptation, sayed2013diffusion, nassif2020multitask, nedic2009distributed, koloskova2019decentralized,zgheib2025sparsity}. By keeping data on local devices and exchanging only model updates during communication with neighboring agents, decentralized learning reduces the risk of sharing sensitive data. 
Various cooperation rules have been developed to enable effective collaboration among agents, including incremental~\cite{cattivelli2011analysis,bertsekas1997new}, consensus~\cite{sardellitti2010fast,Karconvergence2011}, and diffusion approaches~\cite{sayed2014adaptation,lopes2008diffusion,cattivelli2009diffusion,chen2012diffusion}. In this work, we adopt diffusion strategies because of their scalability, robustness and ability to support continuous adaptation and learning. Furthermore, when constant step-sizes are employed to enable continuous adaptation, these strategies have been shown to have superior stability and performance compared to consensus-based implementations~\cite{tu2012diffusion,sayed2014adaptation,sayed2013diffusion,sayed2014diffusion}.

In \emph{continual learning}, on the other hand, the goal is to emulate how humans adapt to new experiences while retaining previously acquired knowledge. Also known as \emph{incremental learning} or \emph{life-long learning}, this approach addresses the challenge of learning from a sequence of evolving tasks while retaining performance on previously learned tasks~\cite{wang2024comprehensive, mai2022online}. Its effectiveness has been proven in various domains, including autopilot systems~\cite{mirza2022efficient}, emotion recognition~\cite{thuseethan2022emotion}, and Internet of Things (IoT)~\cite{liu2021internet}, where data arrive sequentially and their distributions change over time. 
Such distribution shifts are also  common in medical image analysis, where  the underlying data distribution evolves due to changes in patient populations, disease characteristics, imaging devices, and acquisition protocols. Consequently, clinical models must incorporate new diagnostic information without compromising previously acquired knowledge, since retraining from scratch is often infeasible and naive fine-tuning may result in catastrophic forgetting~\cite{wu2024continual}. Similar challenges arise in embodied vision applications, where mobile robots and wearable assistants  interact with dynamic environments. As novel object categories emerge and past observations reoccur over  time, these systems must continuously acquire new visual knowledge while maintaining existing recognition capabilities. This  highlights the importance of continual learning for long-term autonomous adaptation~\cite{Wang2021wanderlust}.

\begin{figure*}[t]
\begin{center}
\includegraphics[scale=0.29]{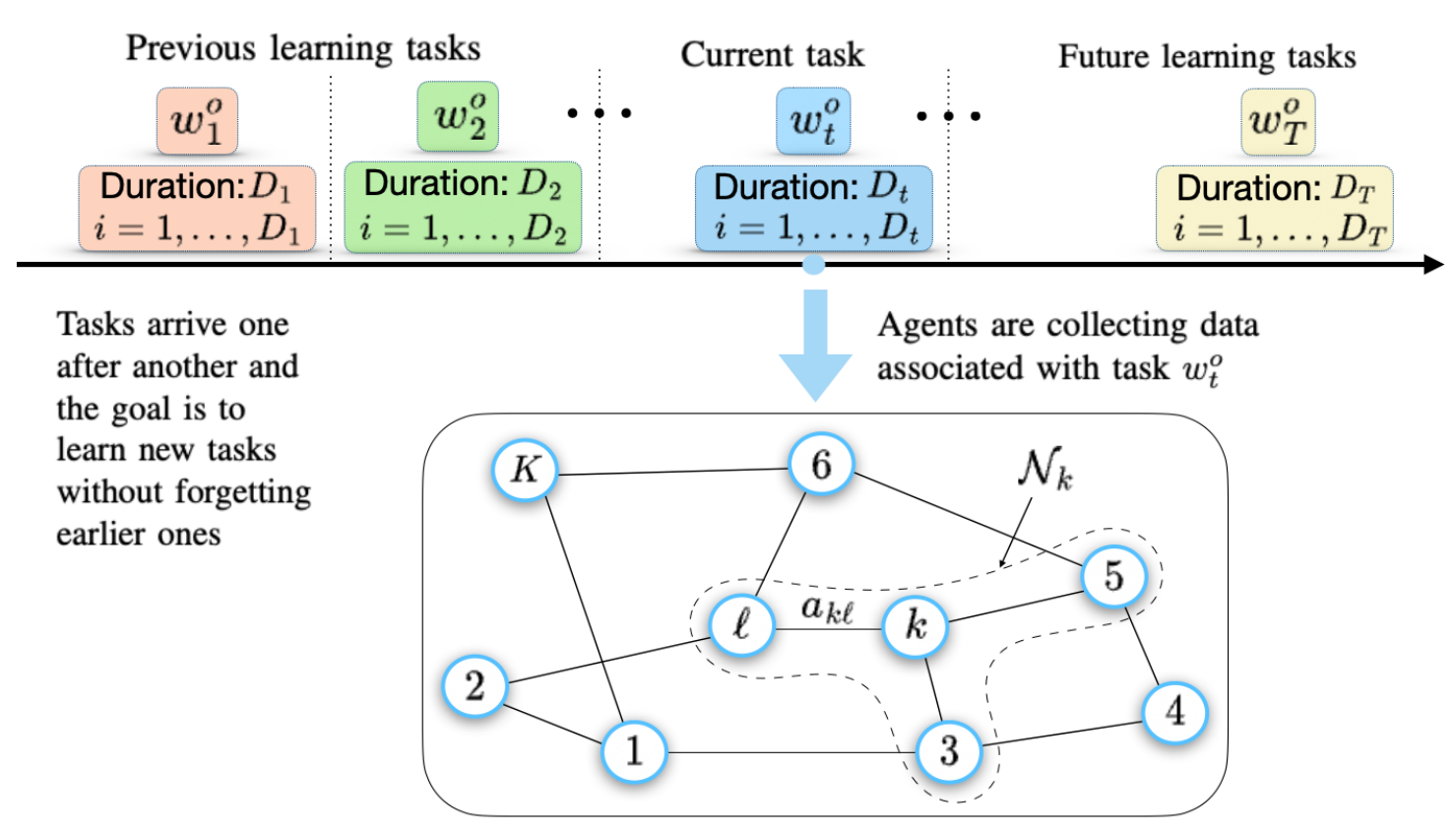} \qquad \quad
\includegraphics[scale=0.4]{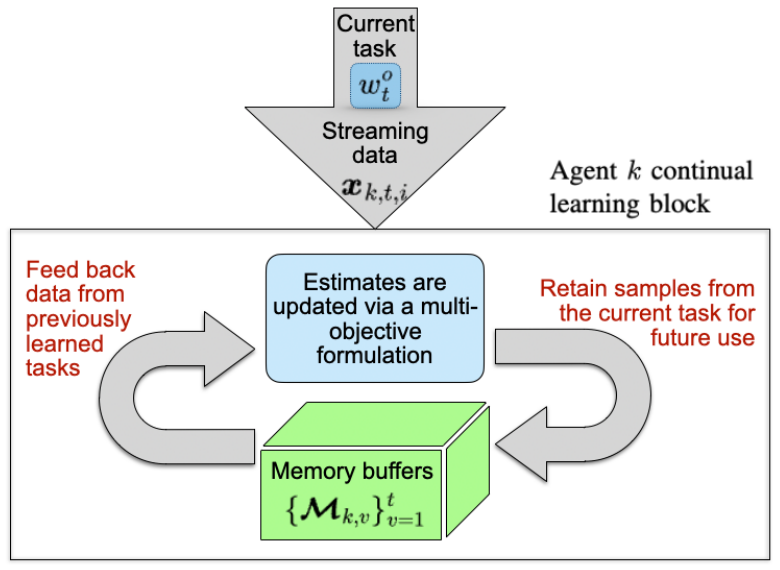}
 \caption{Continual learning over multi-agent network. \emph{(Left)} A network of $K$ agents receives a sequence of tasks over time. The objective is for the agents to learn each new task as it arrives while retaining performance on previously learned tasks. \emph{(Right)} Continual learning block associated with agent $k$, where memory buffers retain data samples from previously learned tasks (which are replayed to reduce forgetting). Agent $k$ seeks to minimize the multi-objective risk~\eqref{eq: jmulti}, while simultaneously exploiting information from its neighboring agents $\ell\in\cN_k$ through the minimization of the global objective in~\eqref{eq: global}.}
  \label{fig: intro}
\end{center}
\end{figure*}

 Continual learning aims to overcome a major limitation of traditional learning systems, namely, \emph{catastrophic forgetting}, where learning new tasks (or models) often degrades performance on previously learned tasks~\cite{ kirkpatrick2017overcoming, nguyen2019toward}. It seeks to balance the \emph{stability-plasticity} trade-off, where stability\footnote{Throughout this paper, the term stability is used in two senses: $i)$  the ability of the system to retain previously acquired knowledge, and $ii)$ the ability of the algorithm to avoid divergence in the mean-square-error sense. The intended meaning should be clear from the context in which the term is employed.} refers to the ability to retain performance on previously learned tasks, and plasticity represents the capacity to rapidly adapt to new tasks~\cite{mermillod2013stability}. Strategies for mitigating catastrophic forgetting are generally grouped into  three categories: $i)$ regularization, $ii)$ architecture, and $iii)$ rehearsal. Regularization-based methods incorporate explicit regularization terms that constrain changes to parameters important for previously learned tasks~\cite{zenke2017continual,chaudhry2018riemannian}. Notable examples include Elastic Weight Consolidation (EWC)~\cite{kirkpatrick2017overcoming}, and EWC++, an efficient online variant~\cite{mai2022online}.  
Regularization-based methods preserve privacy by avoiding data storage, but may limit adaptation as constraints accumulate, making it harder to  learn new tasks effectively.
 Architecture-based methods are more appropriate for scenarios where there are no constraint on the model size. They identify and preserve task-specific parameters while expanding the network with new parameter branches to accommodate new knowledge
 ~\cite{yoon2017lifelong, hung2019compacting}.  Rehearsal-based methods store a small subset of past training samples in a memory buffer and update the model using both new data and stored samples. The buffer is continuously updated with incoming mini-batches~\cite{aljundi2019gradient}. Among these methods is the Gradient Projection Memory (GPM) approach,  which allows a neural network to learn new tasks, while constraining gradient updates to remain orthogonal to subspaces associated with previously learned tasks~\cite{saha2021gradient}. Another approach is Gradient Episodic Memory (GEM), which constrains the updates for the current task so that the loss on  previously seen tasks (approximated using samples stored in an episodic-memory) does not increase~\cite{lopez2017gradient,zgheib2026continual}.  
In online continual learning, memory buffer management has been addressed  through gradient-based selection~ \cite{aljundi2019gradient}, reservoir sampling~\cite{isele2018selective}, and exemplar-based methods such as iCaRL \cite{rebuffi2017icarl}. Finally, it is worth noting that, 
although closely related to multitask learning~\cite{nassif2020multitask,zgheib2025sparsity}, continual learning addresses a fundamentally different setting. In multitask learning, task diversity is handled simultaneously, whereas continual learning considers tasks that arrive sequentially over time, requiring the model to adapt continuously while retaining previously acquired knowledge.

In systems where multiple agents or devices are required to perform continual learning, decentralized learning provides a natural framework for exploiting the distributed data available across agents without directly sharing raw data. Through the exchange of local model updates or parameter estimates, agents can collaboratively adapt to evolving data distributions while benefiting from the knowledge acquired by neighboring agents. Despite its potential, this direction has only recently started to receive attention in the literature. For instance,~\cite{choudhary2023codec} introduces CoDeC, a Communication-efficient Decentralized Continual learning framework that combines stochastic gradient descent with gossip averaging to learn from spatially distributed private data, while employing the Gradient Projection Memory (GPM) approach to enable incremental learning. Similarly,~\cite{hellkvist2024distributed} investigates how the COmmunication-efficient COordinate Ascent (CoCoA) framework can support continual learning through both analytical characterization and numerical experiments. This work considers the domain-incremental setting, where tasks share the same structure but differ in their input distributions.

In the current work, we consider  a network of agents  exposed to tasks arriving sequentially over time--see Fig.~\ref{fig: intro} for an illustration. The objective is to perform online continual learning, where each agent observes a stream of data samples and processes each sample only once, except for a limited subset retained in a local memory buffer~\cite{mai2022online}. We formulate in Sec.~\ref{sec: Section 1} a global multi-objective minimization problem  that balances learning new tasks and retaining performance on previous tasks, and propose a diffusion-based method to solve it in a fully decentralized manner. We  analyze in Sec.~\ref{sec: Section 2} its mean-square-error stability to understand how various factors, such as data characteristics, learning step-size, and regularization balance between plasticity and (knowledge) stability, affect the performance of the approach. We investigate in Sec.~\ref{sec: Section 4} the impact of collaboration on the average mean-square deviation (MSD) across tasks. Since the average MSD quantifies the distance between the current estimate and the optimal solutions of current and  past tasks, reducing it improves both adaptation and knowledge retention. We show that collaboration among agents leads to lower average MSD than the non-cooperative setting, where each agent performs continual learning independently. This gain results from exploiting data diversity and collective experience across the network, which improve the stability-plasticity trade-off. 
\emph{To the best of our knowledge, this work is among the first to provide a theoretical analysis of continual learning approaches}. Finally, the simulations in Sec.~\ref{sec: Section 5}  illustrate the theoretical mean-square-error results and show the gains of cooperation in terms of average forgetting and MSD across tasks.  Experiments on the real-world  FEMNIST dataset~\cite{Caldas2018LEAFAB} further illustrate the effectiveness of the approach.

\medskip

\noindent\textbf{Notation.} Vectors are represented as column vectors, with random quantities shown in boldface. Matrices are denoted by uppercase letters, while lowercase letters denote vectors and scalars. We use the symbol $(\cdot)^\top$ to denote matrix
transpose, and the symbol $\otimes$ to denote Kronecker product. The operator $\text{col}\{\cdot\}$ stacks its arguments as column vectors,  the operator $\diag\{\cdot\}$ forms a block-diagonal matrix from its arguments, and the $ M\times M$ identity matrix is denoted by $I_{M}$. The cardinality operator of a set~$X$ is denoted by~$|X|$. We use the indices $k,t,$ and $i$, to denote agents, tasks, and iterations, respectively. 
 \vspace{-0.1cm}
\section{Problem formulation and decentralized continual-learning strategy}
\label{sec: Section 1}

\noindent{\textbf{(Problem setup)}.} We consider a {strongly} connected network of~$K$ agents, {indexed by $\{1,\ldots,K\}$}, aiming to estimate a common task, $w^o_t$, that evolves over time, {giving rise} to a sequence of~$T$ tasks, $\{w^o_1,\ldots,w^o_T\}$. The task boundaries ({\emph{i.e.}}, the {time indices marking transitions between consecutive tasks}) are assumed to be known a priori -- see Fig.~\ref{fig: intro} \emph{(left)}. While operating on the $t$-th task, agent~$k$ is associated with a risk function ${J_{k,t}(w)}: \mathbb{R}^M \rightarrow \mathbb{R}$, assumed to be strongly convex and, therefore, admitting a unique minimizer denoted by $w^o_t$. The risk function $J_{k,t}(\cdot)$ is expressed as the expectation of some loss function $Q_{k,t}(\cdot)$, and is written as $J_{k,t}(w) = \expec Q_{k,t}(w; \bx_{k,t})$, where $\bx_{k,t}$ denotes the random data collected by agent~$k$ when solving task~$t$, and the expectation is computed with respect to the data distribution. When task~$t$ is under consideration, and in a typical decentralized formulation, rather than having each agent solve for $w^o_t$ independently, the agents collaboratively address the following optimization problem~\cite{sayed2013diffusion, sayed2014adaptation,nedic2009distributed,sayed2014diffusion, lopes2008diffusion, cattivelli2009diffusion, chen2012diffusion}:
\vspace{-0.2cm}
\begin{equation}
\label{eq: local minimizers}
w^o_{t}\triangleq \arg\min_{w} \sum_{k=1}^{K} {J_{k,t}(w)}.
\vspace{-0.2cm}
\end{equation}
Its solution is $w^o_t$ since all individual risk functions attain their minimum at the same point. 

\medskip

\noindent\textbf{(Memory buffer construction)}. To mitigate forgetting of previously learned tasks, each agent is equipped with a local memory buffer that stores a subset of previously observed samples for future use -- see Fig.~\ref{fig: intro} \emph{(right)}. That is, while learning task $t$, each data sample $\bx_{k,t,i}$ collected at iteration $i$ by agent $k$ is independently stored with probability $p_{k,t}$. This sample selection mechanism is modeled through a Bernoulli random variable, $\boldsymbol{\alpha}_{k,t}(i)$, associated with sample $\bx_{k,t,i}$, where $\boldsymbol{\alpha}_{k,t}(i)=1$ indicates that the sample is stored in memory, and $\boldsymbol{\alpha}_{k,t}(i)=0$ otherwise. The retained samples collectively form the \emph{local memory buffer} \(\boldsymbol{\mathcal{M}}_{k,t}\)  at agent $k$:
\vspace{-0.1cm}
\begin{equation}
\vspace{-0.1cm}
\boldsymbol{\mathcal{M}}_{k,t}\triangleq\left\{\bx_{k,t,i}|\boldsymbol{\alpha}_{k,t}(i)=1\right\},
\vspace{-0.1cm}
\end{equation}
with size \(\boldsymbol{N}_{k,t}=|\boldsymbol{\mathcal{M}}_{k,t}|\). The retained samples serve as a replay mechanism during the learning of subsequent tasks, $\{w^o_{t+1},\ldots,w^o_T\}$, in order to reduce  forgetting of previous task $w^o_t$. In particular, they are used to construct the empirical risk function~$\boldsymbol{J}^{\text{emp}}_{k,t}(w)$ associated with task $t$, defined as the sample average over the corresponding set $\boldsymbol{\mathcal{M}}_{k,t}$:
\vspace{-0.1cm}
\begin{align}
\vspace{-0.1cm}
\label{eq: emp}
\bJ^{\text{emp}}_{k,t}(w) &\triangleq \frac{1}{p_{k,t}D_t} \sum_{\bx_{k,t,i}\in\boldsymbol{\mathcal{M}}_{k,t}} Q_{k,t}(w;\bx_{k,t,i})\notag\\
&=\frac{1}{p_{k,t}D_t} \sum_{i=1}^{D_t} \boldsymbol{\alpha}_{k,t}(i)Q_{k,t}(w;\bx_{k,t,i}),
\vspace{-0.3cm}
\end{align}
where $D_t$ is the total number of samples observed at agent $k$ for task $t$, assumed to be the same for all agents. The normalization by $p_{k,t}D_t$ ensures that the expectation of the empirical risk $\boldsymbol{J}^{\text{emp}}_{k,t}(w)$ equals the expected risk $J_{k,t}(w)$. Note that both the memory buffer $\boldsymbol{\mathcal{M}}_{k,t}$ and the empirical risk $\boldsymbol{J}^{\text{emp}}_{k,t}(w)$ are random quantities due to the random sample retention process.

\medskip

\noindent{\textbf{(Multi-objective cost formulation)}.} {To balance the plasticity required for adapting to the current task $w^o_t$ with the stability induced by exposure to data from previously learned tasks $\{w^o_1,\ldots,w^o_{t-1}\}$, we associate with agent $k$ a \emph{multi-objective cost} defined as:
\vspace{-0.2cm}
 \begin{equation}
 \vspace{-0.1cm}
\label{eq: jmulti}
 \bJ_{k,t}^{\text{multi}}(w)\triangleq \frac{1}{t} \left(\sum_{v=1}^{t-1} \lambda^{t-v} \bJ^{\text{emp}}_{k,v}(w) + J_{k,t}(w)\right),
 \vspace{-0.1cm}
\end{equation}
where $\lambda\in[0,1] $ is a {\emph{memory-retention}} parameter that controls the trade-off between stability and plasticity. {Larger values of $\lambda$ assign greater importance to previous tasks, while smaller values favor adaptation to new tasks.} The resulting cost~\eqref{eq: jmulti} blends empirical and expected risks. {The factor $\lambda^{t-v} $ introduces an exponential weighting over previous tasks, assigning larger importance to more recent tasks while gradually reducing the influence of older ones.} 

\medskip

\noindent{\textbf{(Collaborative decentralized learning)}.} {To leverage network cooperation, we consider the following minimization problem associated with the aggregate cost, $\bJ_t^{\text{glob}}(w)$, obtained by combining the individual objectives $\{ \bJ_{k,t}^{\text{multi}}(w)\}$ across agents:
\vspace{-0.2cm}
\begin{equation}
\vspace{-0.1cm}
\label{eq: global}
 \arg\min_{w}\left\{\bJ_t^{\text{glob}}(w)\triangleq \sum_{k=1}^{K}  \bJ_{k,t}^{\text{multi}}(w)\right\}.
 \vspace{-0.1cm}
\end{equation}
Problem~\eqref{eq: global} is to be solved in a fully decentralized manner, as agents are only allowed to perform local computations and exchange information with neighboring agents according to a predefined network topology.} We let $\cN_k$ denote the set of neighbors of agent $k$ including $k$ itself, and build upon the adapt-then-combine (ATC)  form of the diffusion strategy~\cite{sayed2013diffusion, sayed2014diffusion, lopes2008diffusion, cattivelli2009diffusion, chen2012diffusion} to derive Algorithm~\ref{alg:dcl} {for solving problem~\eqref{eq: global}}.
 \begin{algorithm}[ht]
    \caption{{Diffusion ATC-based approach for decentralized continual learning}}
    \label{alg:dcl}
    
    \textbf{Initialization:} At each agent $k$, initialize $w_{k,1,0}=0$, learning step-size $\mu$, task durations   $\{D_t\}_{t=1}^T$, combination coefficients $\{a_{k \ell}\}_{\ell\in\cN_k}$, retention parameter $\lambda\in[0,1]$, sampling probabilities $\{p_{k,t}\}_{t=1}^{T}$
    
    \For{$t = 1$ \textbf{to} $T$}{
        \If{$t \neq 1$}{\vspace{1mm}
            $\bw_{k,t,0} = \bw_{k,t-1,D_{t-1}}$\vspace{1mm}
        }
        $\boldsymbol{\mathcal{M}}_{k,t}= [\quad]$ for $k=1, \dots, K$ 
        
        \For{$i = 1$ \textbf{to} $D_t$}{
            \For{$k = 1$ \textbf{to} $K$}{     \vspace{1mm}   
               $ \bullet$ collect data  $\bx_{k,t,i} $ \\ \vspace{1mm}
               $ \bullet$ perform the \textbf{adaptation} step:
                \begin{equation}
                \bpsi_{k,t,i} = \bw_{k,t,i-1} -  \mu   \widehat {\nabla_{w} \bJ_{k,t}^{\text{multi}}}( \bw_{k,t,i-1})\label{eq: adaptation step}
                \end{equation} \\\vspace{1mm}
                where $\widehat {\nabla_{w} \bJ_{k,t}^{\text{multi}}}( \bw_{k,t,i-1})$ is given by~\eqref{eq: stochasmultiobj} \\\vspace{1mm}
                 $ \bullet$ generate \textbf{data storing} variable $\boldsymbol{\alpha}_{k,t}(i)$: 
                 \begin{equation}
                 \text{if }\boldsymbol{\alpha}_{k,t}(i) = 1,\text{then  }\boldsymbol{\mathcal{M}}_{k,t} \leftarrow \boldsymbol{\mathcal{M}}_{k,t}\cup\, \bx_{k,t,i}\label{eq: storing step}
                 \end{equation}
            }
        
        \For{$k = 1$ \textbf{to} $K$}{
           $ \bullet$ perform the \textbf{combination} step:
            \begin{equation}
            \bw_{k,t,i} = \sum_{\ell \in \mathcal{N}_k} a_{k\ell} \bpsi_{\ell,t,i}\label{eq: combination step}
            \end{equation}             }
        }
    }
\end{algorithm}

The \emph{adaptation} step~\eqref{eq: adaptation step}, also referred to as the \emph{self-learning} step, involves performing a stochastic gradient update (using a small positive step-size parameter $\mu$) on the individual multi-objective cost $\bJ_{k,t}^{\text{multi}}(w)$ defined in~\eqref{eq: jmulti}. Its gradient approximation can be expressed as: 
 \vspace{-0.1cm}
 \begin{equation}
 \vspace{-0.1cm}
 \label{eq: stochasmultiobj}
\widehat {\nabla_{w} \bJ_{k,t}^{\text{multi}}}( w)\triangleq  \frac{1}{t} \left(\sum_{v=1}^{t-1} \lambda^{t-v} \widehat {\nabla_{w} \bJ^{\text{emp}}_{k,v}}( w)   +  \widehat {\nabla_{w} J_{k,t}}( w) \right),
\vspace{-0.1cm}
\end{equation}
where $\widehat {\nabla_{w} \bJ^{\text{emp}}_{k,v}}( w)$ and $\widehat {\nabla_{w} J_{k,t}}( w)$ denote the gradient estimates associated with the empirical risk $\bJ^{\text{emp}}_{k,v}( w)$ and  the expected risk $J_{k,t}( w) $, respectively. For the expected risk $J_{k,t}(w) = \expec Q_{k,t}(w; \bx_{k,t})$ (where the data distribution is unknown), a common construction in the stochastic approximation theory is to use the instantaneous approximation~\cite{sayed2014adaptation}:
 \vspace{-0.2cm}
\begin{equation}
\vspace{-0.2cm}
\label{eq: approx}
\widehat {\nabla_{w} J_{k,t}}( w) = \nabla_{w} Q_{k,t}(w; \bx_{k,t,i}),
\end{equation}
where $\bx_{k,t,i}$ denotes the data sample observed by agent $k$ at iteration $i$ while solving task $t$. In this case, the gradient noise $\bs_{k,t,i}(w)$, defined as the difference between the true gradient and its approximation, is expressed as~\cite{sayed2014adaptation}:
 \vspace{-0.1cm}
\begin{equation}
 \vspace{-0.1cm}
\bs_{k,t,i}(w) \triangleq \nabla_{w} J_{k,t}( w) - \widehat {\nabla_{w} J_{k,t}}( w). \label{eq: noise}
\end{equation}
{For the empirical risk $\bJ^{\text{emp}}_{k,v}( w)$ defined in~\eqref{eq: emp}, the memory buffer $ \boldsymbol{\mathcal{M}}_{k,v}$ is formed by randomly retaining samples from the streaming data $\bx_{k,v,i}$ observed while learning task $v$. After the retention process, the memory buffer can be viewed as:
 \vspace{-0.1cm}
\begin{equation}
 \vspace{-0.1cm}
 \boldsymbol{\mathcal{M}}_{k,v}=\{\tilde{\bx}_{k,v,1},\ldots,\tilde{\bx}_{k,v,\boldsymbol{N}_{k,v}}\},
  \vspace{-0.1cm}
\end{equation}
where $ \boldsymbol{N}_{k,v}$ denotes the number of retained samples, and the tilde notation is used to relabel the samples in the  buffer from $1$ to $\boldsymbol{N}_{k,v}$.} To reduce computational complexity, the gradient of $\bJ^{\text{emp}}_{k,v}( w)$ is approximated using a single sample selected uniformly at random from  $\boldsymbol{\mathcal{M}}_{k,v}$:
 \vspace{-0.1cm}
\begin{equation}
 \vspace{-0.1cm}
 {\widehat  {\nabla_{w} \bJ^{\text{emp}}_{k,v}}( w) = \frac{\boldsymbol{N}_{k,v}}{p_{k,v}D_v}\nabla_{w} Q_{k,v}(w;\tilde{\bx}_{k,v,\boldsymbol{n}_i})}, \label{eq:noise2} 
  \vspace{-0.1cm}
 \end{equation}
where $\boldsymbol{n}_i$ is a discrete random variable, uniformly distributed over the index set $\{1,\ldots,\bN_{k,v}\}$, indicating which sample from the set $ \boldsymbol{\mathcal{M}}_{k,v}$ is selected at iteration $i$. The resulting gradient noise is expressed as:
 \vspace{-0.1cm}
 \begin{equation}
  \vspace{-0.1cm}
 {\bs}^{\text{emp}}_{k,v,i}( w) = \nabla_{w} \bJ^{\text{emp}}_{k,v}( w) -  \widehat  {\nabla_{w} \bJ^{\text{emp}}_{k,v}}( w) .  \label{eq: gnoisee}
  \vspace{-0.1cm}
\end{equation} 
The self-learning step~\eqref{eq: adaptation step} produces $\bpsi_{k,t,i}$, an intermediate estimate of the solution to problem~\eqref{eq: global} at agent $k$, iteration $i$, while solving task $t$. 

 For the \emph{combination} step~\eqref{eq: combination step}, also referred to as the \emph{social learning} step, agent $k$ aggregates the intermediate estimates $\{\bpsi_{\ell,t,i}\}$ received from its neighboring agents $\ell \in \mathcal{N}_k$, using non-negative combination coefficients~$\{a_{k\ell}\}$. If we collect these coefficients into the $K\times K$ combination matrix $A=[a_{k\ell}]$, then this matrix is chosen to satisfy the following conditions~\cite{sayed2014adaptation}:
  \vspace{-0.2cm}
\begin{align}
 \vspace{-0.1cm}
\label{eq: conditionsA}
& a_{k\ell} \geq 0, ~~ A \mathds{1}_K=\mathds{1}_K,~~\mathds{1}_K^{\top} A=\mathds{1}_K^{\top}, ~~ a_{k\ell} =0 \text{ if } \ell \notin \mathcal{N}_k,
 \vspace{-0.5cm}
\end{align}
where $\mathds{1}_K$ denotes the $K$-dimensional vector of ones. The conditions in~\eqref{eq: conditionsA} are equivalent to requiring $A$ to be doubly-stochastic and to have the same sparsity structure   as the communication topology. Common choices for the combination coefficients include the Laplacian and Metropolis rules~\cite{sayed2014diffusion}.

{Between the adaptation and the combination steps, the observed data sample is stored in the memory buffer of agent~$k$ if the data storing variable $\boldsymbol{\alpha}_{k,t}(i)=1$ -- see step~\eqref{eq: storing step}.}
 \vspace{-0.3cm}
 \section{Network mean-square-error stability} 
\label{sec: Section 2}
In this section, we analyze the mean-square-error stability of the decentralized continual {learning} Algorithm~\ref{alg:dcl} when applied to solve~\eqref{eq: global}. Although Algorithm~\ref{alg:dcl} follows the standard form of diffusion ATC for minimizing aggregate sum of individual risks, its analysis cannot rely on existing results in the literature of diffusion adaptation~\cite{sayed2014adaptation}. The main difficulty arises from the randomness of the empirical-risk component in $ \bJ_{k,t}^{\text{multi}}(w)$ induced by the random sample selection process. In standard analyses, the mean-square error is typically studied through quantities of the form, $\expec\|w^o-\bw_{k,i}\|^2$, where the global target minimizer $w^o$ is deterministic. In the present setting, however, the target optimizer itself becomes random since it depends on the randomly retained samples. Consequently, the mean-square-error analysis must be revisited to explicitly account for this additional source of randomness.
 \vspace{-0.3cm}
\subsection{Conditions on the individual risks, losses, and gradient noise processes}
We {begin} by introducing the following assumptions on the individual cost functions $J_{k,t}(\cdot)$, the loss functions $Q_{k,v}(\cdot)$ associated with previous tasks $\{v=1,\ldots,t-1\}$, and the gradient noise processes $ \bs_{k,t,i}(\cdot)$ defined by~\eqref{eq: noise}. {These assumptions are standard in the literature on expected risk~\cite{sayed2014adaptation} and empirical risk~\cite{yuan2016stochastic} minimization, which is particularly relevant to our setting since the multi-objective formulation~\eqref{eq: jmulti} combines both types of risks.}
\begin{assumption}{(Conditions on risks and losses).} \label{ass: boundedhess}
For each agent $k$, the expected risk $J_{k,t}(w)$ associated with the current task $t$ is assumed to be twice-differentiable and strongly convex. Moreover, the Hessian matrix function $H_{k,t}(w)\triangleq  \nabla_{w} ^2J_{k,t}(w)$ is assumed to satisfy the following condition:
 \vspace{-0.1cm}
\begin{equation}
 \vspace{-0.1cm}
0 < \nu_{k,t} I_M \leq H_{k,t}(w) \leq \delta_{k,t} I_M,
\end{equation}
for some $0<\nu_{k,t}\leq \delta_{k,t}$. For every agent $k$, previous task $\{v=1,\ldots,t-1\}$, and sample index~$n \in \{1,\ldots,D_v\}$, the loss $Q_{k,v}(w;\bx_{k,v,n})$ is twice-differentiable and convex with respect to $w$. In addition, there exists a constant $\delta_{k,v,n} \ge 0$ such that, for any $w_1,w_2 \in \mathbb{R}^M$,
 \vspace{-0.1cm}
\begin{align}
 \vspace{-0.1cm}
&\Vert \nabla_{w} Q_{k,v}(w_1; \bx_{k,v,n}) -  \nabla_{w} Q_{k,v}(w_2; \bx_{k,v,n}) \Vert \notag \\
& \leq \delta_{k,v,n} \Vert w_1 -w_2\Vert .
 \vspace{-0.2cm}
\end{align}
We define $\delta_{k,v} \triangleq \max_{1 \le n \le D_v} \delta_{k,v,n}$, so that all sample-loss gradients $ \nabla_{w} Q_{k,v}(\cdot; \bx_{k,v,n}) $  for task $v$ at agent $k$ are $\delta_{k,v}$-Lipschitz continuous.
\hfill\qed
\end{assumption}
\begin{assumption}{(Conditions on gradient noise processes).}
\label{ass: gradient noise}
For each agent $k$, the gradient noise process defined in~\eqref{eq: noise} is assumed to satisfy:
 \vspace{-0.1cm}
\begin{align}
 \vspace{-0.1cm}
&\expec\left[\bs_{k,t,i}(\bw_{k,t,i-1}) \mid\bw_{k,t,i-1}\right] = 0, \label{eq: fisrtmomentemp}\\
&\expec\left[\| \bs_{k,t,i}(\bw_{k,t,i-1})\|^2  \mid\bw_{k,t,i-1}\right] \notag\\
&\qquad\leq \beta_{s,k,t}^2 \| w^o_t-\bw_{k,t,i-1} \|^2 + \sigma_{s,k,t}^2, 
\label{eq: condition on variance of expected noise}
 \vspace{-0.1cm}
\end{align}
 \vspace{-0.1cm}
for some $\beta_{s,k,t}^2 \geq 0$,  and $\sigma_{s,k,t}^2 \geq 0$.
\hfill\qed
\end{assumption}
\begin{assumption}{(Bounded gradient moments at the minimizer). } \label{ass:finite-moment-Q} 
For every agent $k$, previous task $\{v=1,\ldots,t-1\}$, and sample index~$n \in \{1,\ldots,D_v\}$, the gradient of the loss function $Q_{k,v}(w^o_v;\bx_{k,v,n})$ evaluated at the task minimizer $w^o_v$ has bounded second-order moment, namely,
  \vspace{-0.1cm}
\begin{equation}
 \vspace{-0.1cm}
  \mathbb{E}\Big[ \big\|\nabla_w Q_{k,v}\big(w^o_v; \bx_{k,v,n}\big)\big\|^2 \Big]  \le  G_{k,v}^2  < \infty,
  \label{eq:ass-finite-moment-Q}
   \vspace{-0.1cm}
\end{equation}
where the expectation is taken with respect to the randomness of the data.
 \hfill\qed
\end{assumption}
\begin{lemma}{(Characterization of the empirical risk and gradient noise process).} \label{lemma1}
Under Assumption~\ref{ass: boundedhess}, the random Hessian matrix $\bH^{\emph{emp}}_{k,v}(w)$ of the empirical risk $\bJ^{\emph{emp}}_{k,v}(w)$, defined in~\eqref{eq: emp}, satisfies:
 \vspace{-0.1cm}
\begin{equation}
 \vspace{-0.1cm}
0  \leq \bH^{\emph{emp}}_{k,v}(w) \leq \boldsymbol{\delta}^{\emph{emp}}_{k,v} I_M,   \label{eq: hessian of empirical risk}
 \vspace{-0.1cm}
\end{equation}
where $ \boldsymbol{\delta}^{\emph{emp}}_{k,v}$ is the non-negative random scalar given by:
 \vspace{-0.1cm}
\begin{equation}
\label{eq: nu and delta emp}
\boldsymbol{\delta}^{\emph{emp}}_{k,v} \triangleq \frac{\boldsymbol{N}_{k,v}}{p_{k,v}D_v}\delta_{k,v}.
\end{equation}
 Moreover, the empirical gradient noise process~$\bs^{\emph{emp}}_{k,v,i}(w)$ defined in~\eqref{eq: gnoisee} satisfies: 
  \vspace{-0.1cm}
\begin{align}
&\expec\left[\bs^{\emph{emp}}_{k,v,i}(\bw_{k,t,i-1})  \mid\bw_{k,t,i-1},\boldsymbol{\mathcal{M}}_{k,v}\right] = 0, \label{eq: fisrtmoment} \\
&\expec\left[\| \bs^{\emph{emp}}_{k,v,i}(\bw_{k,t,i-1})\|^2\mid\bw_{k,t,i-1},\boldsymbol{\mathcal{M}}_{k,v}\right] \notag\\
&\qquad \leq (\boldsymbol{\beta}^{\emph{emp}}_{s,k,v})^2 \| w^o_v - \bw_{k,t,i-1} \|^2 + (\boldsymbol{\sigma}^{\emph{emp}}_{s,k,v})^2,\label{eq: condition on variance of empirical noise}
 \vspace{-0.1cm}
\end{align}
where $(\boldsymbol{\beta}^{\emph{emp}}_{s,k,v})^2$ and  $(\boldsymbol{\sigma}^{\emph{emp}}_{s,k,v})^2$ are non-negative quantities given by:
 \vspace{-0.1cm}
\begin{align}
(\boldsymbol{\beta}^{\emph{emp}}_{s,k,v})^2& \triangleq 2 \delta_{k,v}^2\left(\frac{\boldsymbol{N}_{k,v}}{p_{k,v}D_v}\right)^2,\\
 (\boldsymbol{\sigma}^{\emph{emp}}_{s,k,v})^2 &\triangleq  2 \left(\frac{\boldsymbol{N}_{k,v}}{p_{k,v}D_v}\right)^2 \max_{1 \le n \le D_v} \left\| \nabla_w Q_{k,v}\big(w^o_v; \bx_{k,v,n}\big) \right\|^2.\label{eq: defintion of sigma carre skv}
 \end{align}
  \vspace{-0.1cm}
\end{lemma}
 \vspace{-0.5cm}
\begin{proof}
See Appendix~\ref{appendix: proof of lemma 1}.
\end{proof}
Lemma~\ref{lemma1} and Assumption~\ref{ass: boundedhess} together imply the  strong convexity of the individual hybrid costs $\{ \bJ_{k,t}^{\text{multi}}(w)\}$ defined in~\eqref{eq: jmulti}, and thus the aggregate objective in~\eqref{eq: global} has a unique solution, denoted by $\bw^\star_t$: 
 \vspace{-0.1cm}
\begin{equation}
 \vspace{-0.1cm}
\label{eq: wstar}
  \bw_t^\star \triangleq \arg\min_w \bJ^{\text{glob}}_t(w).
   \vspace{-0.1cm}
\end{equation}
The random solution $\bw^\star_t$ depends on the previously retained samples and generally differs from~$w^o_t$.
 \vspace{-0.3cm}
\subsection{Mean-square-error stability analysis}
 \vspace{-0.1cm}
We begin by deriving the error recursion of Algorithm~\ref{alg:dcl} when applied to solve problem~\eqref{eq: global}. Let $ \tilde{\bw}_{k,t,i}$ represent the error vector at node $k$ (while solving task $t$ under the continual learning formulation~\eqref{eq: global}) at iteration $i$, which is given by:
 \vspace{-0.1cm}
\begin{equation}
 \vspace{-0.1cm}
\tilde{\bw}_{k,t,i} \triangleq \boldsymbol{w}_{t}^\star - \bw_{k,t,i}.
 \vspace{-0.1cm}
\end{equation}
By using  \eqref{eq: stochasmultiobj} and the mean-value theorem~\cite[pp. 24]{polyak1987introduction}, and by introducing the following definition:
 \vspace{-0.1cm}
 \begin{equation}
  \vspace{-0.1cm}
\bs_{k,t,i}^{\text{multi}}(w) \triangleq \nabla_{w} \bJ_{k,t}^{\text{multi}}( w) - \widehat {\nabla_{w} \bJ_{k,t}^{\text{multi}}}( w),  \label{eq: noisemulti}
 \vspace{-0.1cm}
\end{equation}
we can express the stochastic gradient vectors appearing in the self-learning step in Algorithm \ref{alg:dcl} as follows: 
 \vspace{-0.1cm}
\begin{align}
 \vspace{-0.1cm}
\label{eq:inequality}
\widehat {\nabla_{w} \bJ_{k,t}^{\text{multi}}}( \bw_{k,t,i-1})= &- \bH_{k,t,i-1}^{\text{multi}} \tilde{\bw}_{k,t,i-1} +\notag\\
&  \boldsymbol{b}^{\text{multi}}_{k,t} - \bs_{k,t,i-1}^{\text{multi}}( \bw_{k,t,i-1}),
 \vspace{-0.1cm}
\end{align}
where
 \vspace{-0.1cm}
\begin{equation}
\vspace{-0.1cm}
\label{eq: hermitian}
 \bH_{k,t,i-1}^{\text{multi}} \triangleq \int_{0}^{1} \nabla_{w} ^2 \bJ_{k,t}^{\text{multi}}( \boldsymbol{w}^\star_{t} + \tau \bwt_{k,t,i-1})d\tau ,
 \vspace{-0.1cm}
 \end{equation}
and 
 \begin{equation}
 \label{eq: bt}
\boldsymbol{b}^{\text{multi}}_{k,t}\triangleq \nabla_{w} \bJ_{k,t}^{\text{multi}}( \boldsymbol{w}_{t}^\star). 
\end{equation}
By using~\eqref{eq:inequality}, and by introducing the following vector and matrix definitions:
 \vspace{-0.1cm}
\begin{align}
\bcwt_{t,i}& \triangleq \col \{ \bwt_{1,t,i}, ..., \bwt_{K,t,i}\}, \label{eq:weighterror} \\ 
\bcH^{\text{multi}}_{t,i-1}&\triangleq \diag \{ \bH_{1,t,i-1}^{\text{multi}}, ..., \bH_{K,t,i-1}^{\text{multi}}\},\\
\bcB_{t,i-1}& \triangleq  \cA (I_M - \mu \bcH^{\text{multi}}_{t,i-1}),\\
\boldsymbol{b}^{\text{multi}}_{t}  &\triangleq  \col \{ \boldsymbol{b}^{\text{multi}}_{1,t}, ..., \boldsymbol{b}^{\text{multi}}_{K,t}\},  \label{eq: btm}\\
\bs^{\text{multi}}_{t,i} & \triangleq  \col \{ \bs_{1,t,i-1}^{\text{multi}}(\bw_{1,t,i-1}), ..., \bs_{K,t,i-1}^{\text{multi}}(\bw_{K,t,i-1})\},
\end{align} 
which collect quantities across the network, 
we can show that the network weight error vector $\bcwt_{t,i}$, defined in \eqref{eq:weighterror}, evolves according to the following recursion: 
 \vspace{-0.1cm}
\begin{equation}
 \vspace{-0.1cm}
\label{eq: final}
\bcwt_{t,i}= \bcB_{t,i-1}\bcwt_{t,i-1} - \mu \cA\bs^{\text{multi}}_{t,i} + \mu \cA \boldsymbol{b}^{\text{multi}}_{t} ,
 \vspace{-0.1cm}
\end{equation}
where:
 \vspace{-0.1cm}
\begin{equation}
\cA\triangleq A\otimes I_M,
 \vspace{-0.1cm}
\end{equation} and where we used the fact that:
\begin{equation}
 \vspace{-0.1cm}
\cA \bcw_t^\star=\bcw_t^\star,
 \vspace{-0.1cm}
\end{equation}
with $\bcw_t^\star\triangleq \col\{\boldsymbol{w}^\star_t, \ldots, \boldsymbol{w}^\star_t\}$.

The stability of recursion \eqref{eq: final} is facilitated by transforming it into an appropriate basis given by the Jordan decomposition of the combination matrix $\cA$. Under the conditions in~\eqref{eq: conditionsA}, the Jordan decomposition can be expressed as~\cite[Chapter 9]{sayed2014adaptation}:
 \vspace{-0.1cm}
\begin{equation}
 \vspace{-0.1cm}
    \mathcal{A} \triangleq \mathcal{V}_{\epsilon} \Lambda_{\epsilon}\mathcal{V}_{\epsilon}^{-1},
\end{equation}
where:
 \vspace{-0.1cm}
\begin{equation}
\begin{split}
    \Lambda_{\epsilon} &= 
    \begin{bmatrix} 
        I_M & 0 \\ 
        0 & \mathcal{J}_{\epsilon} 
    \end{bmatrix}, \quad
    \mathcal{V}_{\epsilon} =
    \begin{bmatrix}
        \frac{1}{\sqrt{K}}(\mathds{1}_K \otimes I_M) & \mathcal{V}_{R,\epsilon} 
    \end{bmatrix}, \\
    \mathcal{V}_{\epsilon}^{-1} &=  
    \begin{bmatrix}
        \frac{1}{\sqrt{K}}(\mathds{1}_K^{\top} \otimes I_M)\\
        \mathcal{V}_{L,\epsilon}^{\top}
    \end{bmatrix},
    \end{split}
     \vspace{-0.1cm}
\end{equation}
and $\cJ_{\epsilon}$ is a Jordan matrix with the eigenvalues (which may be complex but have magnitude less than one) on the diagonal and small $\epsilon > 0$ on the super-diagonal.
Multiplying both sides of \eqref{eq: final} from the left by $\cV_{\epsilon}^{-1}$ and introducing the transformed iterates and variables: 
 \vspace{-0.1cm}
\begin{align}
    \mathcal{V}_{\epsilon}^{-1} \bcwt_{t,i} &=
    \begin{bmatrix}
         \frac{1}{\sqrt{K}}(\mathds{1}_K^{\top}  \otimes I_M) \bcwt_{t,i} \label{eq: transformedvariables} \\
        \mathcal{V}_{L,\epsilon}^{\top} \bcwt_{t,i} 
    \end{bmatrix}  \triangleq
    \begin{bmatrix}
       \boldsymbol{\cwb}_{t,i} \\
        \bcwc_{t,i}
    \end{bmatrix}, \\
   \mathcal{V}_{\epsilon}^{-1} \mathcal{A} \bs_{t,i}^{\text{multi}} &= 
    \begin{bmatrix}
         \frac{1}{\sqrt{K}}(\mathds{1}_K^{\top}  \otimes I_M)  \bs_{t,i} ^{\text{multi}}\\
         \cJ_{\epsilon}\mathcal{V}_{L,\epsilon}^{\top} \bs_{t,i}^{\text{multi}}
    \end{bmatrix}
    \triangleq
    \begin{bmatrix}
      \overline{\bs}_{t,i}^{\text{multi}} \\
       \widecheck{\bs}_{t,i}^{\text{multi}}
    \end{bmatrix}, \label{eq: eqfinalnoise}\\
   \mathcal{V}_{\epsilon}^{-1} \mathcal{A}  b_{t}^{\text{multi}}&= 
    \begin{bmatrix}
         \frac{1}{\sqrt{K}}(\mathds{1}_K^{\top}  \otimes I_M)\boldsymbol{b}_{t}^{\text{multi}} \\
         \cJ_{\epsilon}\mathcal{V}_{L,\epsilon}^{\top}\boldsymbol{b}_t^{\text{multi}}
    \end{bmatrix}
    \triangleq 
    \begin{bmatrix}
        0 \\
   \widecheck {\boldsymbol{b}}_t^{\text{multi}}
    \end{bmatrix}, \label{eq: bias}
\end{align}
where the $0$ entry in~\eqref{eq: bias} arises because $\bw_t^\star$ is the unique minimizer of the aggregate objective in~\eqref{eq: global}, we obtain:
 \vspace{-0.1cm}
\begin{align}
    \boldsymbol{\cwb}_{t,i}  &= (I - \mu \bcD_{11,t,i-1})   \boldsymbol{\cwb}_{t,i-1} - \mu \bcD_{12,t,i-1}  \bcwc_{t,i-1} -  \mu \overline{\bs}_{t,i}^{\text{multi}} \label{eq: new W}\\
     \bcwc_{t,i} &= (\mathcal{J}_{\epsilon} - \mu \bcD_{22,t,i-1}) \bcwc_{t,i-1}  - \mu \bcD_{21,t,i-1} \boldsymbol{\cwb}_{t,i-1} -    \notag\\
     &  \qquad \mu \widecheck{\bs}_{t,i}^{\text{multi}} + \mu \widecheck {\boldsymbol{b}}_t^{\text{multi}},
     \label{eq: Neww}
\end{align}
where $\bcD_{11,t,i-1}=  \frac{1}{{K}} \sum_{k=1}^K \bH_{k,t,i-1}^{\text{multi}},$ $\bcD_{12,t,i-1}=  \frac{1}{\sqrt{K}}(\mathds{1}_K^{\top} \otimes I_M)  \bcH_{t,i-1}^{\text{multi}} \mathcal{V}_{R,\epsilon}$, $\bcD_{21,t,i-1} =  \mathcal{J}_{\epsilon} (\mathcal{V}_{L,\epsilon})^{\top} \bcH_{t,i-1}^{\text{multi}}  \left(\frac{1}{\sqrt{K}}(\mathds{1}_K \otimes I_M)\right)$, and $\bcD_{22,t,i-1} =  \mathcal{J}_{\epsilon} (\mathcal{V}_{L,\epsilon})^{\top} \bcH_{t,i-1}^{\text{multi}} \mathcal{V}_{R,\epsilon}.$

Note that the strong convexity of the hybrid costs $\{ \bJ_{k,t}^{\text{multi}}(w)\}$ defined in~\eqref{eq: jmulti} guarantees that   $ \bcD_{11,t,i-1}$ is positive definite.
\begin{theorem} \label{theorem2}(Network mean-square-error stability).  Consider a network of $K$ agents running Algorithm \ref{alg:dcl} with a  matrix $A$ satisfying~\eqref{eq: conditionsA}. Under Assumptions \ref{ass: boundedhess}, \ref{ass: gradient noise}, and \ref{ass:finite-moment-Q}, when the network is solving task~$t$, it is mean-square-error stable for sufficiently small step-sizes. In particular, the mean-square behavior of the transformed error  vectors is characterized by:
 \vspace{-0.1cm}
\begin{equation}
\begin{bmatrix}
\mathbb{E} \| \boldsymbol{\cwb}_{t,i} \|^2 \\
\mathbb{E} \|  \bcwc_{t,i}\|^2
\end{bmatrix}
\preceq 
\Gamma
\begin{bmatrix}
\mathbb{E} \| \boldsymbol{\cwb}_{t,i-1}  \|^2 \\
\mathbb{E} \|  \bcwc_{t,i-1}\|^2
\end{bmatrix}
+
\begin{bmatrix}
e \\
e + f
\end{bmatrix} \label{eq: finfin}
 \vspace{-0.1cm}
\end{equation}
where $\Gamma$ is a $2\times 2$ matrix with the following asymptotic structure:
 \vspace{-0.1cm}
\begin{align}
& \Gamma \triangleq
\begin{bmatrix}
1-\mu \sigma_{11} + O(\mu^2) & O(\mu) \\
O(\mu^2) & \rho(\mathcal{J}_{\epsilon}) + \epsilon + O(\mu^2)
\end{bmatrix},
\end{align}
and $e$ and $f$ are non-negative constants given by:
 \vspace{-0.1cm}
\begin{align}
 \vspace{-0.1cm}
& e \triangleq   \mu^2 v_1^2 \sum_{k=1}^K \expec \left[(\boldsymbol{{\sigma}}^{\emph{multi}}_{s,k,t})^2\right],\\
&f  \triangleq \frac{3 \mu^2 \expec \Big\| \boldsymbol{\widecheck {b}}_t^{\emph{multi}} \Big\|^2}{1 - \rho(\mathcal{J}_{\epsilon}) - \epsilon}, 
 \vspace{-0.1cm}
\end{align}
with {$v_1\triangleq \|  \mathcal{V}_{\epsilon}^{-1} \|$,}  $\sigma_{11}$ is some positive constant independent of $\mu$,
 \vspace{-0.1cm}
\begin{align}
 \vspace{-0.1cm}
& \left(\boldsymbol{{\sigma}}^{\emph{multi}}_{s,k,t}\right)^2 \triangleq  \frac{2(t-1)}{t^2} \sum_{v=1}^{t-1}\lambda^{2(t-v)}\Big(4 \left(\frac{\delta_{k,v}}{p_{k,v}}\right)^2 \Vert w^o_v - \boldsymbol{w}_{t}^\star \Vert^2  \notag \\
& \qquad \qquad + (\boldsymbol{\sigma}^{\emph{emp}}_{s,k,v})^2 \Big) + \frac{2}{t^2} \left(2 \beta_{s,k,t}^2 \Vert w^o_t -  \boldsymbol{w}_{t}^\star \Vert^2 + \sigma_{s,k,t}^2 \right). \label{eq: sigmamulti}
 \vspace{-0.2cm}
\end{align}
and $\boldsymbol{\sigma}^{\emph{emp}}_{s,k,v}$ is given by~\eqref{eq: defintion of sigma carre skv}. Using Assumption~\ref{ass:finite-moment-Q}, it can be shown that $e$ and $f$ are  $O(\mu^2)$. Moreover, for sufficiently small $\mu$, the matrix $\Gamma$ is stable, and  by iterating~\eqref{eq: finfin}, we obtain:  
 \vspace{-0.1cm}
\begin{align}
\limsup_{i \to \infty} 
\begin{bmatrix}
\mathbb{E} \| \boldsymbol{\cwb}_{t,i}\|^2 \\
\mathbb{E} \| \bcwc_{t,i} \|^2
\end{bmatrix}
&\preceq 
(I - \Gamma)^{-1}
\begin{bmatrix}
e \\
e + f
\end{bmatrix}\notag\\
&=
\begin{bmatrix}
\frac{\mu v_1^2 \sum_{k=1}^K \expec \left[(\boldsymbol{{\sigma}}^{\emph{multi}}_{s,k,t})^2\right]}{\sigma_{11}}+O(\mu^2) \\
O(\mu^2)
\end{bmatrix} \label{eq: finaliteartion}
\end{align}
from which we conclude that $\limsup_{i \to \infty} \mathbb{E} \| \boldsymbol{\cwb}_{t,i} \|^2 = O(\mu)$ and $\limsup_{i \to \infty} \mathbb{E} \| \bcwc_{t,i} \|^2 = O(\mu^2)$. Therefore,
 \vspace{-0.1cm}
\begin{equation}
 \vspace{-0.1cm}
\limsup_{i \to \infty} \mathbb{E} \| \bcwt_{t,i}  \|^2= O(\mu). 
 \vspace{-0.1cm}
\end{equation}
 \vspace{-0.1cm}
 \end{theorem}
  \vspace{-0.5cm}
 \begin{proof}
See Appendix~\ref{appendix: proof of theorem1}.
\end{proof}
 \vspace{-0.4cm}
\section{Impact of collaboration on continual learning metrics}
 \vspace{-0.1cm}
\label{sec: Section 4}
Previous distributed learning studies   have shown that, by exploiting data diversity, cooperation among agents enhances the network's ability to optimize the global objective~\cite{sayed2014adaptation}. In this section, we examine how collaboration influences continual learning from the perspective of stability-plasticity trade-off. In particular, we study its impact on the average mean-square deviation (MSD) across tasks, defined as~\cite{lin2023theory}: 
 \vspace{-0.1cm}
\begin{equation}
 \vspace{-0.1cm}
\label{eq: averagemsd}
\text{MSD}_{\text{av},t}(i) \triangleq \frac{1}{t} \sum_{v=1}^{t} \left[  \frac{1}{K} \sum_{k=1}^{K}  \expec\|w^o_{v} - \bw_{k,t,i}\|^2 \right],
 \vspace{-0.1cm}
\end{equation}
where~$w_v^o$ denotes the minimizer associated with task~$v$, and $\bw_{k,t,i}$ represents the iterate at agent~$k$ while solving task~$t$, at iteration~$i$. This metric captures both stability, reflected by maintaining low error on previously encountered tasks, and plasticity, reflected by adapting to new tasks. In particular, stability corresponds to small values of  $\expec\|w^o_{v} - \bw_{k,t,i}\|^2$, $v=1,\ldots,t-1$, while plasticity corresponds to small values of $ \expec\|w^o_{t} - \bw_{k,t,i}\|^2$. 

We begin by introducing some definitions  that will be used in the subsequent  analysis. For each agent~$k$, we define the multi-objective cost function:
 \vspace{-0.1cm}
\begin{equation}
 \vspace{-0.1cm}
J_{k,t}^{\text{multi}}(w) = \frac{1}{t}  \left(\sum_{v=1}^{t-1} \lambda^{t-v} J_{k,v}(w) + J_{k,t}(w)\right), \label{eq: jmult}
 \vspace{-0.1cm}
\end{equation}
where~$J_{k,v}(w) = \expec Q_{k,v}(w; \bx_{k,v})$. The newly introduced cost $J_{k,t}^{\text{multi}}(w)$ is the expected version of the random empirical risk $\bJ_{k,t}^{\text{multi}}(w)$ defined in~\eqref{eq: jmulti}. Following~\eqref{eq: global}, we define the (deterministic) aggregate global cost as:
 \vspace{-0.2cm}
\begin{equation}
 \vspace{-0.2cm}
 J_t^{\text{glob}}(w) \triangleq \sum_{k=1}^{K}  J_{k,t}^{\text{multi}}(w),\label{eq: J global t deterministic}
  \vspace{-0.1cm}
\end{equation}
which is strongly convex, with a unique minimizer given by:
 \vspace{-0.1cm}
\begin{equation}
 \vspace{-0.1cm}
w_t^{\star} \triangleq  \arg\min_{w} J_t^{\text{glob}} (w).
 \vspace{-0.1cm}
\end{equation}

We start decomposing the average $\text{MSD}_{\text{av},t}(i)$ in~\eqref{eq: averagemsd}. To this end, we introduce and subtract $w_t^\star$, and apply Jensen's inequality  $\|a+b\|^2\le 2\|a\|^2+2\|b\|^2$ (which holds for any two vectors $a$ and $b$ of compatible sizes) to write:
 \vspace{-0.1cm}
\begin{equation}
 \vspace{-0.1cm}
\expec\|w_v^o-\bw_{k,t,i}\|^2
\le 2\|w_v^o-w_t^\star\|^2 + 2\expec\|w_t^\star-\bw_{k,t,i}\|^2. \label{eq:msd_split}
 \vspace{-0.1cm}
\end{equation}
Substituting \eqref{eq:msd_split} into the definition of $\text{MSD}_{\text{av},t}(i)$, 
and adding and subtracting $\bw_t^\star$, the unique minimizer of~\eqref{eq: global}, in the term $\expec\|w_t^\star-\bw_{k,t,i}\|^2$, yields
 \vspace{-0.1cm}
\begin{align}
\text{MSD}_{\text{av},t}(i) &\le \frac{2}{t}\sum_{v=1}^t \|w_v^o-w_t^\star \|^2 + 4\expec\|w_t^\star-\bw_t^\star  \|^2 \notag\\
& \quad +  \frac{4}{K}\sum_{k=1}^K \expec\|\bw_t^\star -\bw_{k,t,i} \|^2 . 
 \vspace{-0.1cm}
\end{align}
The first term, $ \frac{2}{t}\sum_{v=1}^t \|w_v^o-w_t^\star \|^2 $, represents an intrinsic bias that depends on the similarity between tasks. 
Since consecutive tasks in  continual learning  are typically assumed  similar, this term is generally expected to remain small. It 
 is independent of buffer sampling and stochastic optimization. 
The third term, $\expec\|\bw_t^\star -\bw_{k,t,i} \|^2$, is characterized by Theorem~\ref{theorem2}: 
 \vspace{-0.1cm}
\begin{equation}
 \vspace{-0.1cm}
\limsup_{i \to \infty}  \expec \Vert \boldsymbol{w}_{t}^\star - \bw_{k,t,i} \Vert^2 = O(\mu), \quad k=1, \ldots, K.\label{eq: small contribution}
 \vspace{-0.1cm}
\end{equation}
While many terms contribute to the constant multiplying $\mu$ (as it can be observed from~\eqref{eq: sigmamulti} and~\eqref{eq: finaliteartion}), result~\eqref{eq: small contribution} ensures that the contribution of the third term to the average MSD in steady-state remains small in the small step-size regime, i.e., when $\mu\rightarrow 0$. Consequently, we can write:
\begin{align}
&\limsup_{i \to \infty} \text{MSD}_{\text{av},t}(i) \notag\\
&\le \frac{2}{t}\sum_{v=1}^t \|w_v^o-w_t^{\star}\|^2 + 4\expec\|\bw_t^\star -w_t^{\star}\|^2  + O(\mu).
\label{eq:msd-final}
\end{align}

We still need to analyze the second term, $\expec\|\bw_t^\star -w_t^{\star}\|^2$, which, as we will see, depends on the probability of 	retaining samples and on the collaborative continual learning framework. To characterize the influence of cooperation, we analyze this term in both collaborative and non-collaborative settings. In non-collaborative settings, each agent $k$ seeks to minimize its own multi-objective cost $\bJ_{k,t}^{\text{multi}}(w)$ defined in~\eqref{eq: jmulti} independently of the other agents. Let~$\bw_{k,t}^\star$ denote the unique minimizer of $\bJ_{k,t}^{\text{multi}}(w)$, i.e.,
\begin{equation}
\label{eq: local minimizers no collaboration}
\bw_{k,t}^\star\triangleq \arg\min_{w}\bJ_{k,t}^{\text{multi}}(w).
\end{equation}
Because the memory buffers differ across agents, the local minimizers~$\bw_{k,t}^\star$ are generally distinct. Moreover, the minimizer~$\bw_t^\star$ of the global objective, $\bJ_t^{\text{glob}}(w)$ in~\eqref{eq: global}, depends on the information retained throughout the network and generally differs from the local ones~$\bw_{k,t}^\star$. Similar observations hold for the deterministic minimizer $w_{k,t}^\star$, the unique minimizer of $J_{k,t}^{\text{multi}}(w)$ in~\eqref{eq: jmult}, and $w_t^\star$, the unique minimizer of $J_{t}^{\text{glob}}(w)$ in~\eqref{eq: J global t deterministic}. This is because, although for each $v$ the  costs $\{J_{k,v}(w)\}_{k=1}^K$ have a common minimizer $w^o_v$, the corresponding weighted sums $\sum_{v=1}^t\alpha_v J_{k,v}(w)$ generally do not admit a common minimizer $\forall$ $k=1, \ldots,K$. To facilitate  comparison between collaborative and non-collaborative approaches, and to ensure that the first term on the RHS of~\eqref{eq:msd-final} is the same in both approaches, we focus in the following on situations where the minimizers $\{w_{k,t}^\star\}_{k=1}^K$ coincide, thereby ensuring that both approaches have the same global minimizer $w^\star_t$.

 To that end, consider a mean-square-error (MSE) network. When solving task~$t$, agent $k$ is observing a streaming sequence $\{\bd_{k,t}(i), \bu_{k,t,i}\}$ that follows a linear regression model:
\begin{equation}
\label{eq: linear data model}
\bd_{k,t}(i) = \bu_{k,t,i}^{\top}w_{t}^o + \bz_{k,t}(i) ,
\end{equation}
where~$w_{t}^o$ is some $M \times 1$ vector to be estimated by agent~$k$ and~$\bz_{k,t}(i)$ denotes a zero-mean measurement noise. We associate with agent $k$ an MSE cost given by~\cite{sayed2013diffusion}:
\begin{equation}
\label{eq: MSE}
J_{k,t}(w)= \frac{1}{2} \expec \left( \bd_{k,t}(i) - \bu_{k,t,i}^{\top}w \right)^2 ,
\end{equation}
which is minimized at $w_{t}^o$. We assume that regressor covariance matrices $\expec \bu_{k,t,i}\bu_{\ell,t,i}^{\top} = R_{u,k,t}$ are  identical across all agents. As a result, the local minimizers $\{w_{k,t}^\star\}_{k=1}^K$ coincide with the global minimizer~$w_t^\star$ of the global cost~$J_t^{\text{glob}}(w)$, i.e., $w_{1,t}^\star=\ldots=w_{K,t}^\star=w_t^\star$. 

Having identified situations in which the collaborative and non-collaborative strategies can be fairly compared, we now show that the second term on the RHS of~\eqref{eq:msd-final} decreases under multi-agent collaborative continual learning. The following lemma establishes this result by comparing the  two mean-square distances:
\begin{equation}
\underbrace{\frac{1}{K}\sum_{k=1}^K\expec \|\bw_{k,t}^\star-w_t^\star\|^2}_{\text{non-collaborative approach}} \quad \text{and} \quad \underbrace{ \expec \|\bw_{t}^\star-w_t^\star\|^2}_{\text{collaborative approach}}
\end{equation}
which quantify the variance induced by the sample storage process and, consequently, the accuracy with which the non-collaborative and collaborative memory-based objectives approximate the minimizer of the deterministic global cost.

\begin{lemma}{(Benefit of collaboration).} \label{theoremcollab}
 Suppose Assumptions \ref{ass: boundedhess}, \ref{ass: gradient noise}, and \ref{ass:finite-moment-Q}, hold. Let the sampling probabilities be uniform across the agents and the tasks, i.e.,  $p=p_{k,t}$ for all $t,k$, and let all agents perform the same number of iterations per task, i.e., $D=D_t$ for all $t$. It can be shown that:
 \begin{align}
 \expec \|\bw_{k,t}^\star-w_t^\star\|^2 & \le O\left(\frac{1}{Dp}\right), \quad k=1, \ldots, K, \label{eq: noncollab} \\
  \expec \|\bw_{t}^\star-w_t^\star\|^2 & \le   O\left(\frac{1}{KDp}\right). \label{eq: benefit}
 \end{align}
\end{lemma}
\begin{proof}
See Appendix~\ref{appendix: proof of theoremcollab}.
\end{proof}

The established bounds~\eqref{eq: noncollab} and~\eqref{eq: benefit} show that the error induced by memory sampling decreases as the average number of retained samples, $Dp$, increases. The second bound~\eqref{eq: benefit}, on the other hand, reveals an additional factor~$\frac{1}{K}$, which arises from cooperation among the $K$ agents. This factor reflects a variance-reduction effect: by aggregating information stored across multiple memory buffers, the network obtains a better approximation of the deterministic multi-objective cost $J_{t}^{\text{glob}}(w)$ than any individual agent can achieve on its own. Consequently, collaboration reduces the gap between the memory-based solution~$\bw_t^\star$ and the ideal solution~$w_t^\star$, thereby improving stability across tasks and mitigating the loss of information caused by limited local memory. This benefit becomes particularly significant when the local memory buffer is small, i.e., when $Dp$ is small.

\section{Simulations}
\label{sec: Section 5}
In this section, we illustrate the performance of Algorithm~\ref{alg:dcl} in the context of MSE and logistic regression networks~\cite{sayed2014adaptation}.

\subsection{Illustrative simulations}
We start by considering an MSE network with $K=15$ agents with the topology shown in Fig.~\ref{fig: results1} \emph{(left)}, solving~$T=3$ different tasks that are presented sequentially over time. When solving task~$t$, agent $k$ is observing a streaming sequence $\{\bd_{k,t}(i), \bu_{k,t,i}\}$ that follows the linear model in~\eqref{eq: linear data model}. 

The regression and noise processes $\{\bu_{k,t,i}, \bz_{k,t}(i)\}$ are assumed to be zero-mean jointly wide-sense stationary with: $i)$ $\expec \bu_{k,t,i}\bu_{\ell,t,i}^{\top} = R_{u,k,t} =\sigma_{u,k,t}^2I_M> 0$ if $k=\ell$ and $0$ otherwise; $ii)$ $\expec \bz_{k,t}(i)\bz_{\ell,t}(i) = \sigma_{z,k,t}^2$ if $k=\ell$ and $0$ otherwise; and $iii)$ $\bu_{k,t,i}$ and $\bz_{k,t}(i)$ are independent of each other. The variances $\{\sigma_{u,k,t}^2\}$ and $\{\sigma_{z,k,t}^2\}$ are generated from the uniform distributions $\mathcal{U}(1, 1.5)$ and $\mathcal{U}(0.15, 0.25)$, respectively. The models $w_{t}^o$ in~\eqref{eq: linear data model} are $5\times 1$ vectors,  
 evolving across tasks according to $ {w}_{t+1}^o= w_{1}^o + \Delta w_{t+1}$, with $w_{1}^o$ and $\Delta w_{k,t+1}$ drawn from the Gaussian distributions $\cN(0,I_5)$ and $\cN(0,0.1\cdot I_5)$, respectively.  To estimate $w^o_t$,  while retaining performance on previously learned tasks, agents employ Algorithm~\ref{alg:dcl} with MSE costs as defined in~\eqref{eq: MSE} and a matrix $A$ constructed using the Metropolis rule~\cite{sayed2014adaptation}. 
\begin{figure}
\begin{center}
\includegraphics[scale=0.09]{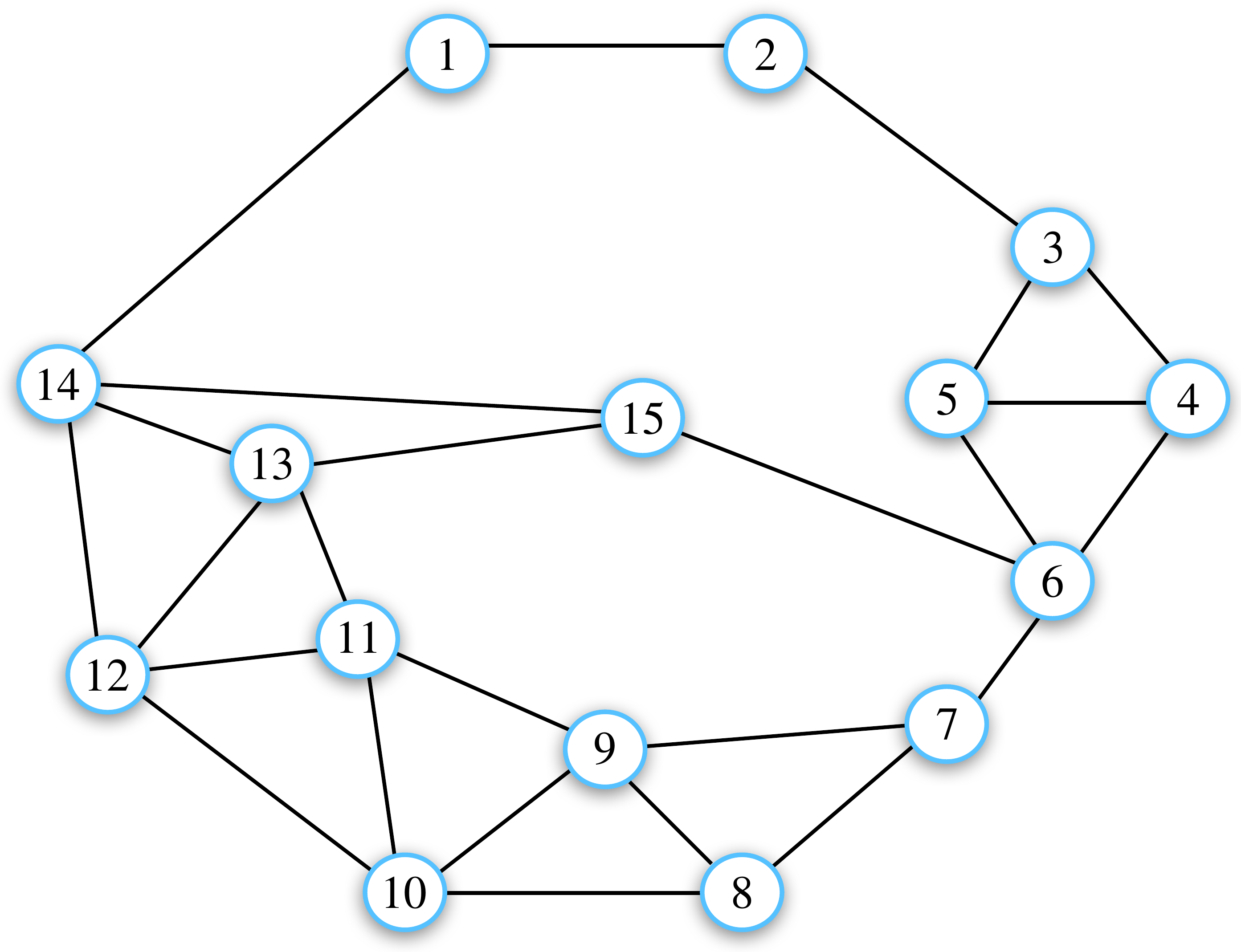}
\includegraphics[scale=0.24]{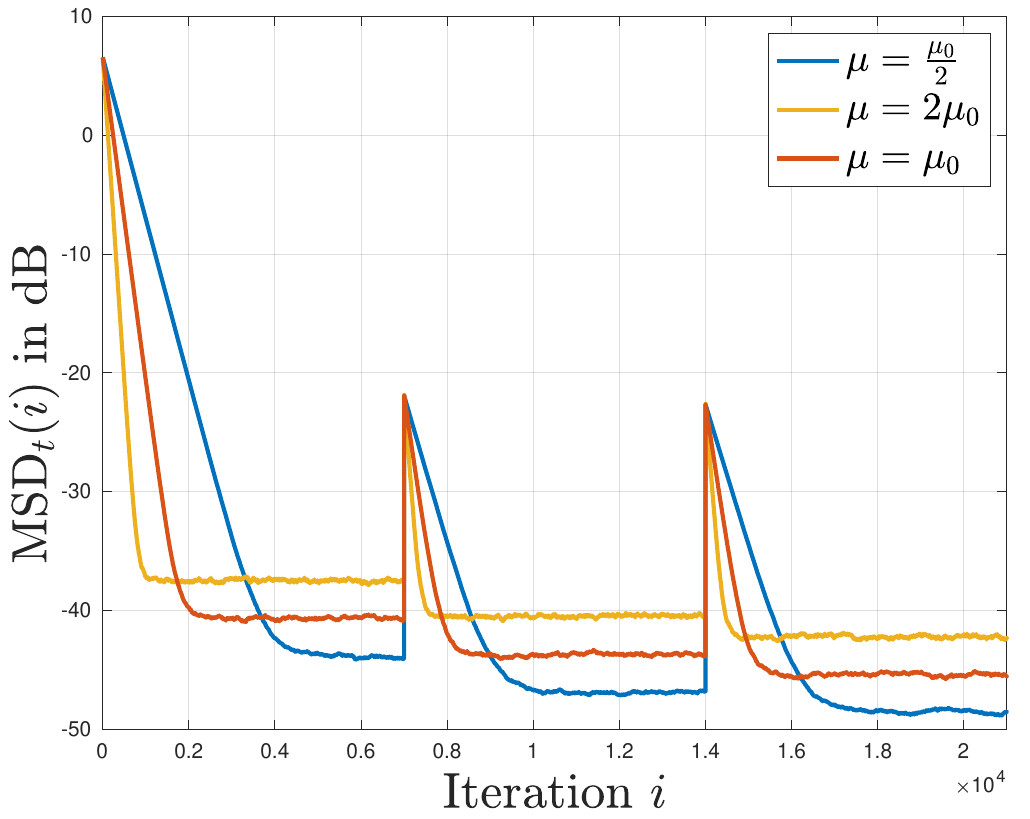}
 \caption{ \emph{(Left)}  Network topology. \emph{(Right)} Performance of  Algorithm~\ref{alg:dcl} w.r.t $\bw_t^\star$, the solution of problem~\eqref{eq: global}  ($\mu_0=0.0025$).} 
   \label{fig: results1}
\end{center}
\end{figure}

\subsubsection{Illustration of the theoretical findings in Theorem~\ref{theorem2}} 
In Fig.~\ref{fig: results1}  \emph{(right)}, we report the evolution of the network mean-square-deviation ($\text{MSD}_t(i)$), defined as:
\begin{equation}
\label{eq: msd w.r. weta}
\text{MSD}_t(i)=\frac{1}{K}\sum_{k=1}^K\expec\|\boldsymbol{w}^\star_{t}-\bw_{k,t,i}\|^2,
\end{equation}
where $\boldsymbol{w}^\star_{t}$ is problem~\eqref{eq: global} solution. The figure shows the concatenated $\text{MSD}_t(i)$ curves ($ \text{MSD}_1(i)$, $ \text{MSD}_2(i)$, and $ \text{MSD}_3(i)$), meaning that the curves are placed sequentially, one after the other, in accordance with the order in which the tasks are learned. The trajectories are evaluated under three different step-sizes $\mu=\{\mu_0/2,\mu_0,2\mu_0\}$. We set  $D_t=7000$ iterations per task, for $t=1,2,3$.  We set the regularization strength to $\lambda=0.9$, and use sampling probabilities ${p}_{k,t}=0.1, \forall t,k$. The results are averaged over~$500$ Monte-Carlo runs. As it can be observed, in steady-state, the network $\text{MSD}$ increases by approximately $3$ dB when $\mu$ goes from $\mu_0$ to $2\mu_0$, illustrating the $O(\mu)$ performance result from Theorem~\ref{theorem2}.
\begin{figure}
\begin{center}
\includegraphics[scale=0.23]{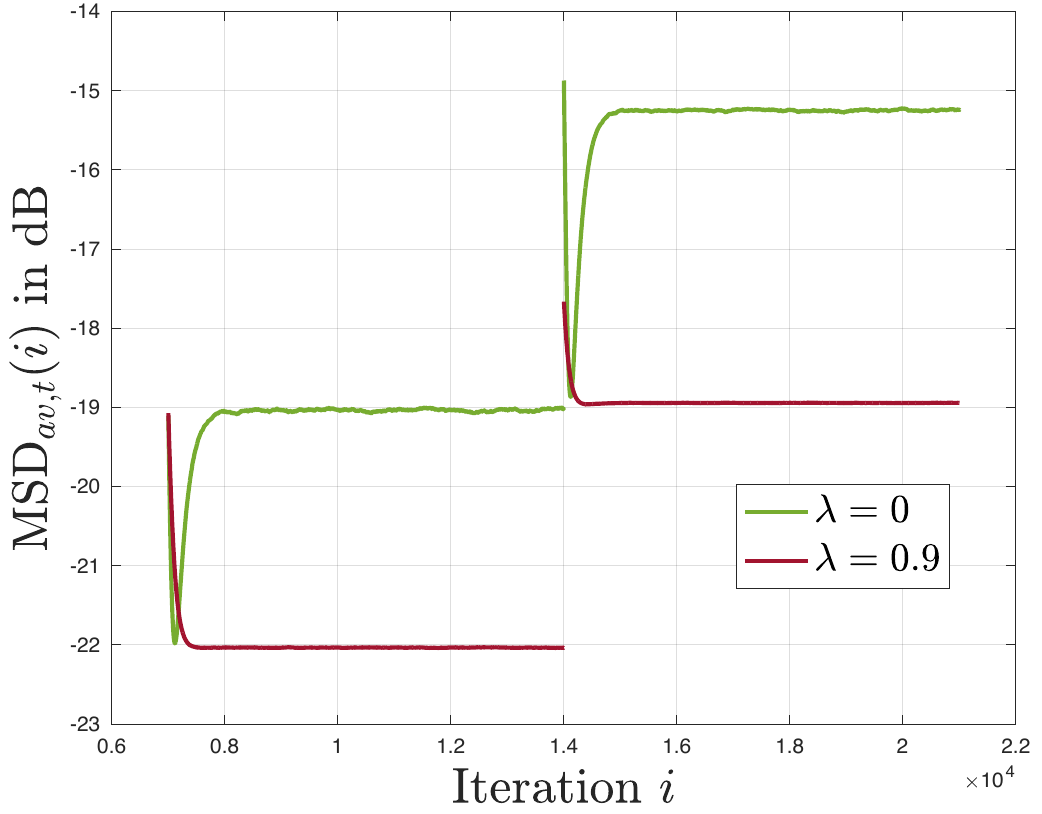}
\includegraphics[scale=0.23]{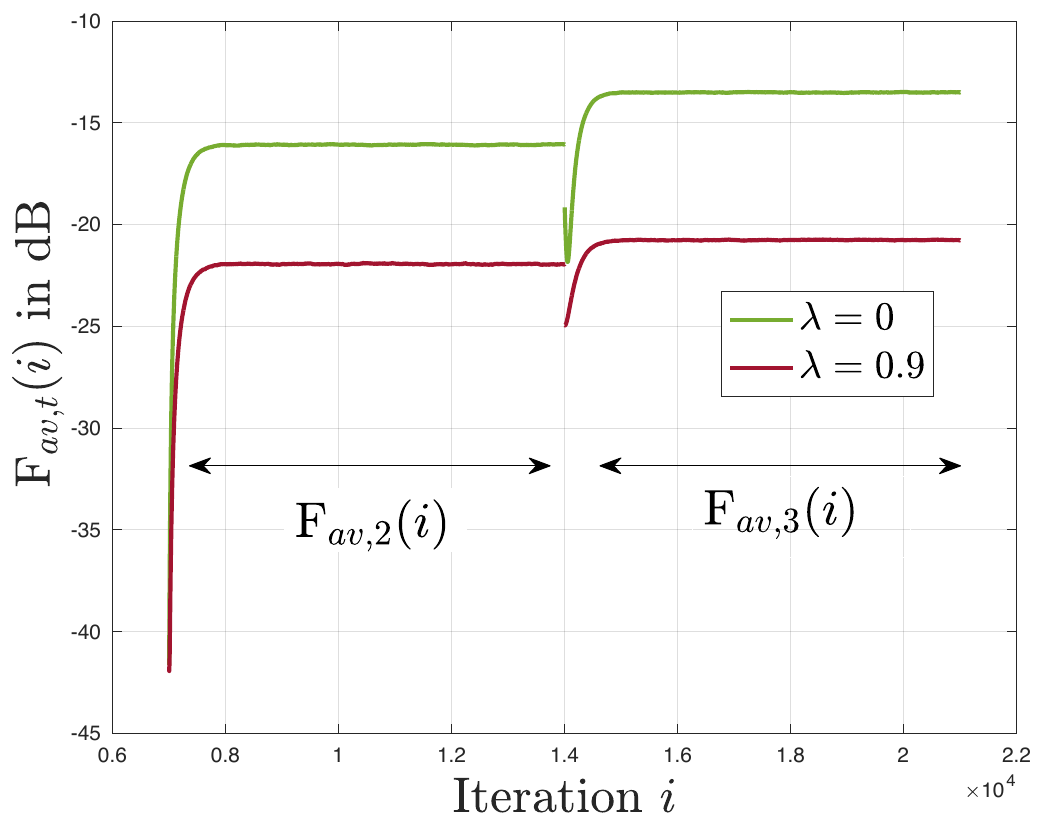}
 \caption{Benefit of continual learning. \emph{(Left)} Average MSD across tasks, defined in~\eqref{eq: averagemsd}.
  \emph{(Right)} Average forgetting error, defined in~\eqref{eq: average forgetting error}.}
  \label{fig: results}
\end{center}
\end{figure}
\subsubsection{Benefit of continual learning}
In Fig.~\ref{fig: results}, we report the evolution of the average mean-square deviation (MSD) across tasks, defined in~\eqref{eq: averagemsd}, as well as the average forgetting defined as~\cite{lin2023theory}: 

\begin{align}
\text{F}_{\text{av},t}(i) &\triangleq \frac{1}{t-1} \sum_{v=1}^{t-1} \Bigg[    \frac{1}{K} \sum_{k=1}^{K} \Big( \expec\|w^o_{v} - \bw_{k,t,i}\|^2  \notag \\
    & \quad - \expec\|w^o_{v} - \bw_{k,v,\infty}\|^2 \Big) \Bigg],\label{eq: average forgetting error}
\end{align}
where $\bw_{k,v,\infty}$ is the steady-state estimate at agent $k$ for task $v$, approximated by averaging the last $200$ iterations during the learning of task $v$. The larger $\text{F}_{\text{av},t}(i) $, the more the network forgets previously learned tasks. We set~$\mu=0.005$. The reported curves in Fig.~\ref{fig: results} correspond to tasks $2$ and $3$. As it can be observed, both metrics improve  when solving the multi-objective continual learning formulation~\eqref{eq: global}, compared to the case where $\lambda=0$, corresponding to learning with no interaction across tasks.
\subsubsection{Benefit of collaborative continual learning}
 In Fig.~\ref{fig:  benefitdl}, we illustrate the benefit of cooperation  by comparing the decentralized  Algorithm~\ref{alg:dcl} with $A=I_K$, corresponding to the non-cooperative solution, against the cooperative solution that employs the Metropolis rule. According to the upper bounds in Lemma~\ref{theoremcollab}, when the number of iterations per task $D$ and the sampling probability $p$ are large, the term $1/Dp$ becomes small and the difference between cooperative and non-cooperative learning becomes less pronounced. To better highlight the advantage of collaboration, we consider a regime where $1/Dp$ is moderate. In particular, we set $D_t=2000$ iterations per task, for $t=1,2,3$, and use sampling probabilities ${p}_{k,t}=0.01$ for all $t$ and $k$. Lower probabilities lead to fewer stored samples and, consequently, reduced memory requirements. In this case, cooperation becomes beneficial since agents can compensate for the reduced local memory by sharing information across the network. As observed in Fig.~\ref{fig: benefitdl}, both the task-averaged MSD and the average forgetting improve under collaboration relative to non-collaborative learning.
 
 \begin{figure}
\begin{center}
\includegraphics[scale=0.23]{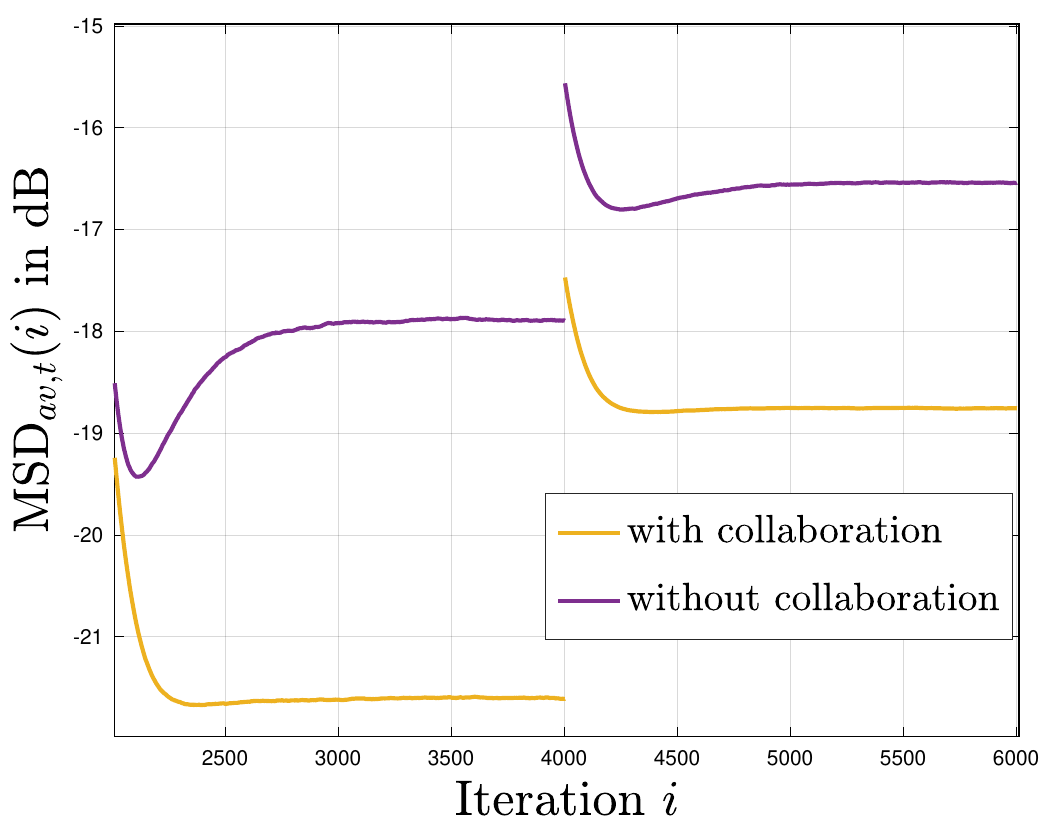}
\includegraphics[scale=0.23]{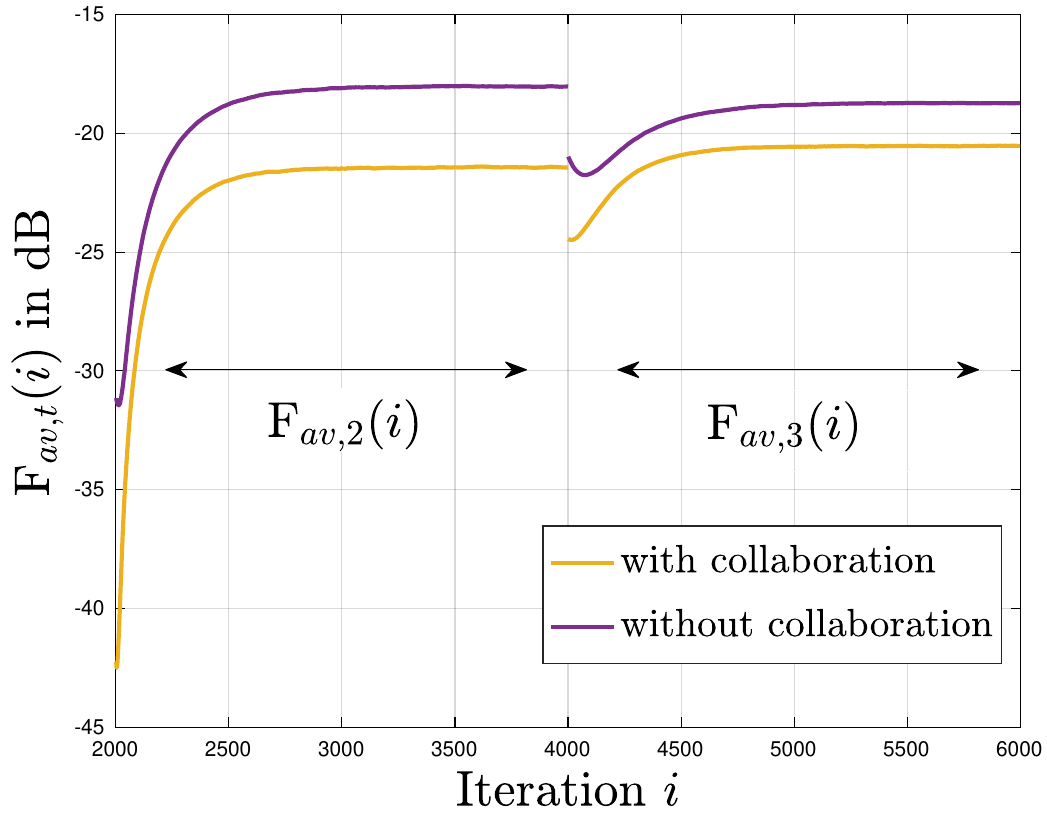}
 \caption{ Benefit of collaborative continual learning. \emph{(Left)} Average MSD across tasks, defined in~\eqref{eq: averagemsd}.  \emph{(Right)} Average forgetting error, defined in~\eqref{eq: average forgetting error}.}
  \label{fig: benefitdl}
\end{center}
\end{figure}

\subsection{Simulation results with FEMNIST}

\begin{table*}[t]
\centering
\label{tab:femnist_results}
\begin{tabular}{c|c c c c c}
\hline\hline
 & $\lambda=0$ & $\lambda=0.1$ & $\lambda=0.9$ \\ 
\hline
Non-collaborative & $61.579 \pm0.032/-6.287 \pm 0.041$ & $64.498 \pm 0.051/-1.909 \pm0.061$  &$ 63.350\pm0.109/-1.643\pm0.115$  \\
\hline
Collaborative & $62.724\pm0.035/-10.529 \pm 0.044$ & $67.841\pm0.035/-2.048\pm0.045$ &$\boldsymbol{68.569\pm0.088/-0.313\pm0.087}$  \\
\hline\hline
\end{tabular}
\caption{ACC/BWT(\%) for the collaborative and non-collaborative settings (sampling probability $p_{k,t}=0.2$).  }
\label{tab:femnist_results}
\end{table*}

In this section, we evaluate the performance of Algorithm~\ref{alg:dcl} on the FEMNIST dataset~\cite{Caldas2018LEAFAB}, a federated version of the Extended MNIST (EMNIST) dataset~\cite{cohen2017EMNIST}, where the samples are partitioned according to the writer of each handwritten character. In this setting, each writer is treated as a separate client, naturally yielding a distributed learning scenario in which every agent has access only to its own local dataset. A  characteristic of this dataset is that the data distribution across clients is not independent and identically distributed (non-IID). Since the samples are grouped by writer, different clients may exhibit distinct writing styles and \emph{different class proportions}. In particular, the number of samples belonging to each character class (e.g., digits such as “$2$”or “$3$” or letters such as “a” or “b”) can vary  from one client to another. 
As learning tasks, we consider binary letter--digit classification problems. Each sample corresponds to a grayscale image of size $28 \times 28 $ representing a handwritten character. To construct a sequence of tasks, we follow the permuted MNIST protocol introduced in~\cite{kirkpatrick2017overcoming}. The first task corresponds to the original images, while the subsequent tasks are generated by applying fixed random permutations to the pixels of the input images. Although the classification objective remains the same, the permutations modify the input representation, resulting in tasks with different input distributions. 

We consider a network of~$K=20$ agents that sequentially solve~$T=10$ tasks. Each  agent holds a local dataset of~$160$ samples per task. Let $ \boldh_{k,t,i} \in \mathbb{R}^M $ denote the feature vector associated with sample~$i$ at agent~$k$ when solving task~$t$, where~$M = 785$ corresponds to the~$28\times 28$ image pixels augmented with a bias term. Let $ \bgamma_{k,t}(i) $ denote the corresponding binary label indicating the class of the character, where $ \bgamma_{k,t}(i)=-1 $ if the sample corresponds to a digit and $ \bgamma_{k,t}(i)=+1$ if it corresponds to a letter. The objective is to learn a classifier that allows each agent to determine whether a given handwritten image represents a digit or a letter based on the observed feature vector $ \boldh_{k,t,i} $. To this end, each agent learns a binary classifier by estimating a parameter vector~$w_{k,t}^o$ that minimizes the regularized logistic risk function~\cite{sayed2014adaptation}:
 \begin{equation}
 \label{eq: logistic}
 J_{k,t}(w_{k,t}) = \expec \ln \left( 1+ e^{-\gamma_{k,t}(i)\boldsymbol{h}_{k,t,i}^{\top}w_{k,t}}\right) + \frac{\rho}{2} \Vert w_{k,t} \Vert^2,
 \end{equation}
 where $\rho > 0$ is a regularization parameter. Once $w_{k,t}^o$ is found, $\widehat{\boldsymbol{\gamma}}_{k,t}(i) = \operatorname{sign}(\boldsymbol{h}_{k,t,i}^{\top} w_{k,t}^o)$ can then be used as a decision rule to classify new features. The cost in~\eqref{eq: logistic} is strongly convex due to the regularization term $\frac{\rho}{2} \Vert w_{k,t} \Vert^2$. 
 
The communication topology between agents is modeled by a weighted graph with a combination matrix~$A$ constructed using the Metropolis rule~\cite{sayed2014adaptation}. We split the dataset into a training set used to learn the decision rule $w_{k,t}^o$, and a test set from which $\widehat{\boldsymbol{\gamma}}_{k,t}(i)$ are generated for performance evaluation. For each task~$t$ and agent~$k$, the local training set consists of~$60$ samples from each class (for a total of $N_{a}=120$ samples), whereas the corresponding test set contains~$20$ samples from each class (for a total of $N_{\text{test}}=40$ samples). Moreover, in our train/test sets construction, some character classes that appear in the training set of a given task and agent may not appear in the corresponding test set, and vice versa.

We assess the generalization performance of Algorithm~\ref{alg:dcl} for different values of the regularization strength $\lambda$ and sampling probabilities $p_{k,v}$. The collaborative setting is compared with the non-collaborative baseline obtained by setting  $A = I_K$. The parameter $\rho$ is fixed to $\rho = 10^{-5}$ and the step-size is set $\mu = 0.01$. 
For each task, the algorithm performs $D = 4000$ stochastic gradient updates, where at every iteration a training sample is drawn uniformly at random with replacement from the local training set of the corresponding agent.  The test (or classification) accuracy at agent~$k$ on task~$v=1,\ldots,t$ after the network has completed the training of task~$t$ is defined as:
\begin{equation}
\label{eq: prederror}
R_{k,t,v}= \expec \left[1 -  \frac{1}{N_{\text{test}}} \sum_{n=1}^{N_{\text{test}}} \mathbb{I} [ \sign ({h}_{k,v,n}^{\top}  \widehat{\bw}_{k,t,\infty}) \neq \gamma_{k,v}(n)] \right],
\end{equation}
where $ \widehat{\bw}_{k,t,\infty}$ is the average of the last $100$ iterates generated by Algorithm~\ref{alg:dcl} for agent $k$ at the end of learning task~$t$, and~$\mathbb{I} [x]$ is the indicator function at $x$, namely, $\mathbb{I} [x]=1$ if $x$ is true and $0$ otherwise. Reported performance metrics are averaged over $100$ Monte-Carlo runs. The network-level accuracy on task~$v$ after learning task~$t$ is then obtained by averaging the individual accuracies across agents:
\begin{equation}
R_{t,v} =\frac{1}{K} \sum_{k=1}^{K} R_{k,t,v}.
\end{equation}
Using this quantity, we evaluate performance using standard continual learning metrics of \emph{average accuracy} and \emph{backward transfer}~\cite{lin2023theory}. The average accuracy is defined as:
\begin{equation}
\label{eq: acc}
\mathrm{ACC} = \frac{1}{T} \sum_{v=1}^{T} R_{T,v}.
\end{equation}
It measures the final average accuracy across all tasks after  last task $T$ has been learned. Backward transfer is defined as
\begin{equation}
\label{eq: forg}
\mathrm{BWT} = \frac{1}{T-1} \sum_{v=1}^{T-1} \left( R_{T,v}-R_{v,v}\right),
\end{equation}
which quantifies the influence of learning new tasks on the performance of previously learned ones. Negative values  indicate forgetting, whereas positive values indicate that learning subsequent tasks improves  performance on earlier~tasks.

\begin{table*}[t]
\centering
\label{tab:results1}
\begin{tabular}{c|c c c c c}
\hline\hline
 & $p_{k,t}=0.08$ & $p_{k,t}=0.2$ & $p_{k,t}=0.4$  \\ 
\hline
Non-collaborative & $62.896\pm0.118/-3.039\pm0.126$ &$ 64.018\pm0.089/-1.837\pm0.108$ & $64.699\pm0.077/\boldsymbol{-1.125\pm0.080}$   \\
\hline
Collaborative & $67.787\pm0.095/-1.716\pm0.106$ &$ 68.288\pm0.067/-1.256\pm0.080$ &$ \boldsymbol{68.485\pm0.049}/-1.128\pm0.056$   \\
\hline\hline
\end{tabular}
\caption{ACC/BWT(\%) for the collaborative and non-collaborative settings (regularization strength $\lambda=0.5$).  }
\label{tab:results1}
\end{table*}

\begin{table*}[t]
\centering
\label{tab:results2}
\begin{tabular}{c|c c c c c}
\hline\hline
 & $p_{k,t}=0.08$ & $p_{k,t}=0.2$ & $p_{k,t}=0.4$  \\ 
\hline
Non-collaborative &$ 61.834\pm0.119/-3.471\pm0.125$ & $63.350\pm0.109/-1.643\pm0.115$ & $64.209\pm0.084/-0.773\pm0.092$    \\
\hline
Collaborative & $67.275 \pm0.129/-1.110\pm0.125$ & $68.569\pm0.088/-0.313\pm0.087$ & $\boldsymbol{69.164\pm0.072/0.111\pm0.066}$  \\
\hline\hline
\end{tabular}
\caption{ACC/BWT(\%) for the collaborative and non-collaborative settings (regularization strength $\lambda=0.9$). }
\label{tab:results2}
\end{table*}

We report in Table~\ref{tab:femnist_results} the values of ACC and BWT for both the collaborative and non-collaborative settings. Results are presented in percentage as mean $\pm 95$\% confidence interval over~$100$ independent trials. The results are shown for different values of $\lambda$, while $p_{k,t}=0.2$ for all tasks~$t$ and agents~$k$. 
The results illustrate the effect of the proposed multi-objective minimization strategy on the continual-learning metrics. When~$\lambda = 0$, the algorithm reduces to standard stochastic gradient algorithm applied only to the current task, without incorporating data from previous tasks. 
In this case, the collaborative setting achieves higher accuracy on the current task, but also exhibits stronger forgetting, which results in more negative BWT values. This behavior occurs because collaboration allows agents to converge more effectively toward the optimal solution of the current task, which can lead to a larger drift away from earlier tasks once a new task is introduced. When the multi-objective formulation is activated~($\lambda>0$), the algorithm incorporates information from previously learned tasks. As a result,  ACC improves in both settings, with a larger gain observed in the collaborative case. At the same time, BWT increases and moves closer to zero as~$\lambda$ grows, indicating a reduction in catastrophic forgetting. These results highlight the classical stability--plasticity trade-off controlled by the parameter~$\lambda$. 

\begin{table*}[t]
\centering
\begin{tabular}{cc|cc}
\hline\hline
$p_{k,t}$ & Memory & Non-collaborative & Collaborative \\
\hline
0.05 & Bernoulli
& $61.045\pm0.130/-4.543\pm0.130$
& $66.434\pm0.153/-1.628\pm0.163$ \\

0.05 & Class-balanced
& $62.430\pm0.152/-2.870\pm0.157$
& $\boldsymbol{66.735\pm0.136/-1.453\pm0.123}$ \\

\hline

0.40 & Bernoulli
& $64.230\pm0.082/-0.801\pm0.091$
& $69.107\pm0.054/0.104\pm0.071$ \\

0.40 & Class-balanced
& $64.302\pm0.094/-0.680\pm0.095$
& $\boldsymbol{69.192\pm0.059/0.141\pm0.071}$ \\

\hline\hline
\end{tabular}
\caption{The impact of memory selection strategies in terms of ACC/BWT(\%) for the collaborative and non-collaborative settings ($\lambda=0.9$).}
\label{tab:results3}
\end{table*}

Tables~\ref{tab:results1} and~\ref{tab:results2}  examine the impact of the storage probability $p_{k,v}$ on ACC and BWT. The results are shown for different values of $p_{k,v}$, assumed identical across tasks~$v$ and agents~$k$, and are presented in percentage as mean~$\pm 95$\% confidence interval over $100$ independent trials. Table~\ref{tab:results1} corresponds to $\lambda=0.5$, while Table~\ref{tab:results2} reports the results for $\lambda=0.9$. As expected, increasing $p_{k,v}$ leads to improved accuracy and reduced forgetting. Moreover, when only a small fraction of past samples is stored, the benefit of decentralized learning becomes clearly visible. 

Let us now modify the strategy used to construct the memory buffers. Although Bernoulli sampling offers a straightforward stochastic mechanism for memory construction, it may produce an imbalanced class distribution in the buffer when the storage probability $p_{k,t}$ is very small.  One way to address this issue is to consider a class-balanced memory construction in which approximately the same number of samples from each class is stored in the replay buffer~\cite{LI2024AdaER,rebuffi2017icarl}. This ensures a balanced representation of classes in the memory. Table~\ref{tab:results3} compares this strategy with Bernoulli sampling for $\lambda = 0.9$ and for two values of $p_{k,t}$, assumed identical across all agents and tasks. Results are reported in percentage as mean~$\pm 95$\% confidence interval over~$100$ independent trials. As it can be observed, the class-balanced memory  reduces forgetting in the non-collaborative case and also provides additional improvements in the collaborative setting. Overall, the impact of the memory-selection strategy on performance is limited. The only noticeable differences arise in the non-collaborative setting when the probabilities $p_{k,t}$ are small.

\section{Conclusion}

We considered decentralized continual learning under streaming data, where multiple agents aim to estimate tasks evolving over time. To address the stability-plasticity dilemma, we formulated a multi-objective optimization problem for estimating  a parameter vector that performs well across  tasks. We then proposed a decentralized approach and established its mean-square-error stability. The analysis reveals that cooperation can provide  benefits in continual learning by improving (knowledge) stability and reducing catastrophic forgetting. Simulations illustrate the effectiveness of the approach.
\appendices
\section{Proof of Lemma~\ref{lemma1}}
\label{appendix: proof of lemma 1}
From Assumption~\ref{ass: boundedhess}, the Hessian matrix function of the  loss function $Q_{k,v}(w;\bx_{k,v,n})$ satisfies $0 \leq   \nabla_{w} ^2Q_{k,v}(w, \bx_{k,v,n}) \leq \delta_{k,v} I_M.$
Applying the above bounds,  term-by-term to the empirical-risk Hessian 
 yields the upper bound:
\begin{equation}
\nabla_w^2 \bJ^{\text{emp}}_{k,v}(w)  \leq  \frac{1}{p_{k,v}D_v}  \delta_{k,v}  \sum_{n=1}^{D_v} \boldsymbol{\alpha}_{k,v}(n)  I_M =  \frac{\boldsymbol{N}_{k,v}}{p_{k,v}D_v}\delta_{k,v} I_M ,
\end{equation}
from which we conclude that~\eqref{eq: hessian of empirical risk} holds. 

To establish~\eqref{eq: fisrtmoment}, we use definition~\eqref{eq:noise2} of the empirical stochastic gradient and write:
 \begin{align}
& \expec\left[ \widehat  {\nabla_{w} \bJ^{\text{emp}}_{k,v}}( \bw_{k,t,i-1})\mid \bw_{k,t,i-1},\boldsymbol{\mathcal{M}}_{k,v}\right] \notag\\
&= \frac{\boldsymbol{N}_{k,v}}{p_{k,v}D_v}  \sum_{n=1}^{D_v} \mathbb{P}(\boldsymbol{n}_i=n \mid\boldsymbol{\mathcal{M}}_{k,v}) \cdot \nabla_w Q_{k,v} (\bw_{k,t,i-1};\bx_{k,v,n}) \notag\\
& = \frac{\boldsymbol{N}_{k,v}}{p_{k,v}D_v}  \sum_{n=1}^{D_v} \frac{\boldsymbol{\alpha}_{k,v}(n)}{\boldsymbol{N}_{k,v}} \nabla_{w} Q_{k,v}(\bw_{k,t,i-1};\bx_{k,v,n}) \notag\\
&= \nabla_{w} \bJ^{\text{emp}}_{k,v}( \bw_{k,t,i-1}),
\end{align} 
from which we conclude that~\eqref{eq: fisrtmoment} holds. Next, to establish the variance bound~\eqref{eq: condition on variance of empirical noise}, we proceed as follows:
\begin{align}
&\expec\left[\| \bs^{\text{emp}}_{k,v,i}(\bw_{k,t,i-1})\|^2\mid\bw_{k,t,i-1},\boldsymbol{\mathcal{M}}_{k,v}\right] \notag\\
&\overset{\text{(a)}}{\leq}  \expec\left[\|  \widehat  {\nabla_{w} \bJ^{\text{emp}}_{k,v}}( \bw_{k,t,i-1})   \|^2\mid\bw_{k,t,i-1},\boldsymbol{\mathcal{M}}_{k,v}\right] \notag\\
& \overset{\text{(b)}}{\leq} \left( \frac{\boldsymbol{N}_{k,v}}{p_{k,v}D_v}\right)^2  \sum_{n=1}^{D_v} \frac{\boldsymbol{\alpha}_{k,v}(n)}{\boldsymbol{N}_{k,v}} \Big[ 2\delta_{k,v}^2 \Vert \bw_{k,t,i-1} - w^o_v\Vert^2  \notag\\
& \qquad +2\left\|  \nabla_{w} Q_{k,v}(w^o_v;\bx_{k,v,n})\right\|^2 \Big],
\end{align}
where in step~(a), we used the fact that for any random variable $\bx$, we have $\expec \Vert \bx - \expec \bx \Vert^2 = \expec \Vert \bx\Vert^2 - \Vert \expec \bx \Vert^2 \leq \expec \Vert \bx\Vert^2,$ and in step~(b), we added and subtracted the vector  $\nabla_{w} Q_{k,v}(w^o_v;\bx_{k,v,n})$, then applied Jensen's inequality~\cite{jensen1906fonctions} together with Assumption~\ref{ass: boundedhess}.

\section{Proof of Theorem~\ref{theorem2}}
\label{appendix: proof of theorem1}

We start by introducing the filtration $\bcF_{t,i-1}$ generated by the iterates and all memory buffers across the network up to time $i-1$ when solving task~$t$. Formally, it is defined as: $ \bcF_{t,i-1} \triangleq \boldsymbol{\sigma}\Big(\{\bw_{\ell,t,j}: \ell=1,\cdots, K, j \leq i-1\}, \{\boldsymbol{\mathcal{M}}_{\ell,v}: \ell=1, \cdots, K, v \leq t\}\Big).$
By the tower property~\cite{williams1991probability}, together with Assumption~\ref{ass: gradient noise} and Lemma~\ref{lemma1}, we obtain:
\begin{equation}
\label{eq: smulti}
\expec\left[\bs_{k,t,i}^{\text{multi}}(\bw_{k,t,i-1}) \mid \bcF_{t,i-1} \right] = 0.
\end{equation}
where~$\bs_{k,t,i}^{\text{multi}}(\bw_{k,t,i-1})$ is defined in~\eqref{eq: noisemulti}. 
We next consider the transformed variables $\boldsymbol{\cwb}_{t,i}$ and $ \bcwc_{t,i} $ defined in~\eqref{eq: new W}, and \eqref{eq: Neww}, respectively. By conditioning both relations on $\bcF_{t,i-1}$, computing the corresponding conditional second-order moments, using~\eqref{eq: smulti} and then taking expectations again, we obtain:
\begin{align}
    &\mathbb{E}\left\|\boldsymbol{\cwb}_{t,i}\right\|^2= \mathbb{E} \Big \| (I_{M} - \mu \bcD_{11,t,i-1}) \boldsymbol{\cwb}_{t,i-1}  \notag\\
    & \qquad  \qquad \quad - \mu \bcD_{12,t,i-1} \bcwc_{t,i-1}  \Big\|^2  + \mu^2 \mathbb{E}\left\|\overline{\bs}_{t,i}^{\text{multi}}\right\|^2, \label{eq: jensenfirstequation} \\
    &\mathbb{E}\|\bcwc_{t,i}\|^2 = \mathbb{E} \Big\| (\mathcal{J}_\epsilon - \mu \bcD_{22,t,i-1}) \bcwc_{t,i-1} - \mu \bcD_{21,t,i-1}  \boldsymbol{\cwb}_{t,i-1}  \notag\\
     & \qquad \qquad \quad  + \mu \widecheck {\boldsymbol{b}}_t^{\text{multi}} \Big\|^2 + \mu^2 \mathbb{E} \left\|\widecheck{\bs}_{t,i}^{\text{multi}} \right\|^2. \label{eq: jensensecond}
\end{align}

The variance~$\mathbb{E}\left\|\boldsymbol{\cwb}_{t,i}\right\|^2$ can be upper bounded according to:
 \vspace{-0.2cm}
\begin{align}
\mathbb{E}\| \boldsymbol{\cwb}_{t,i} \|^2 & \leq (1 - \mu \sigma_{11}) \mathbb{E}\| \boldsymbol{\cwb}_{t,i-1} \|^2 + \frac{\mu \sigma_{12}^2}{\sigma_{11}} \mathbb{E}\| \bcwc_{t,i-1}\|^2   \notag \\
&  + \mu^2 \mathbb{E}\| \overline{\bs}_{t,i}^{\text{multi}} \|^2, \label{eq: equation1}
\end{align}
for sufficiently small step-size~$\mu$. The derivation leading to~\eqref{eq: equation1} is presented in Appendix~D of the Supplementary Material.

Similarly, using the fact that the block-diagonal matrix~$\mathcal{J}_\epsilon$ satisfies~\cite[ p. 516]{sayed2014adaptation} $\| \mathcal{J}_\epsilon \|^2 \leq (\rho(\mathcal{J}_\epsilon) + \epsilon)^2$, together with the fact that~$\rho(\mathcal{J}_\epsilon)\in(0,1)$ and that $\epsilon$ can be selected sufficiently small such that~$ \rho(\mathcal{J}_{\epsilon}) + \epsilon \in (0,1)$, it can be shown that:
\vspace{-0.1cm}
\begin{align}
\mathbb{E} \| \bcwc_{t,i} \|^2 \leq & \left( \rho(\mathcal{J}_{\epsilon}) + \epsilon+ \frac{3 \mu^2 \sigma_{22}^2}{1 - \rho(\mathcal{J}_{\epsilon}) - \epsilon} \right)  \mathbb{E} \| \bcwc_{t,i-1} \|^2  \notag\\
& +  \left( \frac{ 3 \mu^2 \sigma_{21}^2 }{ 1 - \rho(\mathcal{J}_{\epsilon}) - \epsilon} \right) \mathbb{E} \|  \boldsymbol{\cwb}_{t,i-1}\|^2  \notag\\
& +  \left(  \frac{3\mu^2}{ 1 - \rho(\mathcal{J}_{\epsilon}) - \epsilon } \right)  \expec \left\|   \boldsymbol{\widecheck {b}}_t^{\text{multi}}  \right \|^2 + \mu^2 \mathbb{E} \left\| \widecheck{\bs}_{t,i}^{\text{multi}} \right\|^2 . \label{eq: equation22}
\end{align}
Additional details are provided in Appendix~E of the Supplementary Material.

We can also show that the vector~$\boldsymbol{\widecheck {b}}_t^{\text{multi}} $ defined in~\eqref{eq: bias} satisfies: $\expec \left \| \boldsymbol{\widecheck {b}}_t^{\text{multi}} \right\|^2 \leq  (\rho(\mathcal{J}_{\epsilon}) + \epsilon)^2 \left \| \mathcal{V}_{L,\epsilon}^{\top} \right \|^2 \expec \left \|  \boldsymbol {b}_t^{\text{multi}} \right \|^2$.
Before proceeding, we introduce the following Lemma which is necessary to characterize~$\expec \left \|  \boldsymbol {b}_t^{\text{multi}} \right \|^2$.
 \vspace{-0.1cm}
\begin{lemma} \label{lemma2}
Under Assumption~\ref{ass: boundedhess} and \ref{ass:finite-moment-Q}, the optimal solution,~$\bw_t^\star$, satisfies
\begin{equation}
\expec \|\bw_t^\star-w_u^o\|^2 \le C_u, \qquad u\in\{1,\ldots,t\},  \label{eq: finitewstar}
\end{equation}
for some finite constant~$C_u$. Moreover the bias term~$ \boldsymbol {b}_t^{\emph{multi}} $ defined in~\eqref{eq: btm} has a bounded second-order moment. In particular, there exist a finite non-negative constant~$C_2$ such that
\begin{equation}
 \expec \left \|  \boldsymbol {b}_t^{\emph{multi}} \right \|^2   \le C_2.  \label{eq: finiteb}
 \end{equation}
\end{lemma}
\begin{proof}
See Appendix~F of the Supplementary material.
\end{proof}
Using Lemma~\ref{lemma2}, we can conclude that $ \expec \left\| \boldsymbol{\widecheck {b}}_t^{\text{multi}} \right \|^2 = O(1)$.

For the noise terms $\mathbb{E} \left\| \overline{\bs}_{t,i}^{\text{multi}} \right\|^2$ in \eqref{eq: equation1} and $\mathbb{E} \left\| \widecheck{\bs}_{t,i}^{\text{multi}} \right\|^2$ in \eqref{eq: equation22}, we have:
 \vspace{-0.1cm}
\begin{align}
\label{eq: yy}
 \mathbb{E} \left\| \overline{\bs}_{t,i}^{\text{multi}} \right\|^2 + \mathbb{E} \left\| \widecheck{\bs}_{t,i}^{\text{multi}} \right\|^2 &  \overset{\eqref{eq: eqfinalnoise}}{=}  \mathbb{E} \left\|  (\mathcal{V}_{\epsilon})^{-1}  \mathcal{A} \bs_{t,i}^{\text{multi}} \right\|^2 \notag\\
& \leq  v_1^2 \mathbb{E} \left\| \bs_{t,i} ^{\text{multi}} \right\|^2,
\end{align}
where~$v_1$ is a positive constant independent of~$\mu$ and given by~$v_1 \triangleq \| \mathcal{V}_{\epsilon}^{-1} \|$. In terms of the variances of the individual noise processes, $\mathbb{E} \| \bs_{k,t,i}^{\text{multi}} \|^2$, we have $\mathbb{E} \| \bs_{t,i}^{\text{multi}} \|^2 = \sum_{k=1}^K \mathbb{E} \| \bs_{k,t,i}^{\text{multi}} \|^2 $. For each $\bs_{k,t,i}^{\text{multi}}$, we can establish from Assumption~\ref{ass: gradient noise}, Lemma~\ref{lemma1}, and Jensen's inequality that taking expectations again, we can write:
\begin{equation}
  \vspace{-0.1cm}
 \mathbb{E} \| \bs_{k,t,i}^{\text{multi}} \|^2 \leq (\beta^{\text{multi}}_{s,k,t})^2 \expec \| \tilde{\bw}_{k,t,i-1} \|^2 +  \expec [(\boldsymbol{\sigma}^{\text{multi}}_{s,k,t})^2]. \label{eq: 87}
\end{equation}
 \vspace{-0.1cm}
A detailed proof is given in Appendix~G of the Supplementary Material. 

Returning to~\eqref{eq: yy}, and by establishing the following bound:
 \vspace{-0.3cm}
\begin{align}
\mathbb{E} \| \bs_{t,i}^{\text{multi}} \|^2 &= \sum_{k=1}^K \mathbb{E} \| \bs_{k,t,i}^{\text{multi}} \|^2 \notag \\
&\leq  ({{\beta}}^{\text{multi}}_{s,\max,t})^2  \mathbb{E} \| \mathcal{V}_{\epsilon} ( \mathcal{V}_{\epsilon} )^{-1} \bcwt_{t,i-1} \|^2 +     \expec[({\boldsymbol{{\sigma}}}^{\text{multi}}_{s,t})^2]   \notag \\
&\leq ({{\beta}}^{\text{multi}}_{s,\max,t})^2      v_2^2  \left[ \mathbb{E} \|  \boldsymbol{\cwb}_{t,i} \|^2 + \mathbb{E} \| \bcwc_{t,i}  \|^2 \right] \notag\\
& \qquad +  \expec[({\boldsymbol{{\sigma}}}^{\text{multi}}_{s,t})^2], 
\end{align}
where $ ({{\beta}}^{\text{multi}}_{s,\max,t})^2   \triangleq \max_{1 \leq k \leq K} ({\beta}^{\text{multi}}_{s,k,t})^2 $, $({\boldsymbol{{\sigma}}}^{\text{multi}}_{s,t})^2 \triangleq  \sum_{k=1}^K (\boldsymbol{{\sigma}}^{\text{multi}}_{s,k,t})^2$, and $v_2 \triangleq \| \mathcal{V}_{\epsilon} \|$, we conclude that:
 \vspace{-0.1cm}
\begin{align}
& \mathbb{E} \left\| \overline{\bs}_{t,i}^{\text{multi}} \right\|^2 +  \mathbb{E} \left\| \widecheck{\bs}_{t,i}^{\text{multi}} \right\|^2  \notag\\
& \leq  v_1^2 v_2^2  ({{\beta}}^{\text{multi}}_{s,\max,t})^2    \left[ \mathbb{E} \|  \boldsymbol{\cwb}_{t,i} \|^2  + \mathbb{E} \| \bcwc_{t,i}  \|^2 \right]  \notag\\
& +  v_1^2  \expec[({\boldsymbol{{\sigma}}}^{\text{multi}}_{s,t})^2] .
\label{eq: finalso}
\end{align}
 \vspace{-0.4cm}
 
Using the above bound in \eqref{eq: equation1} and \eqref{eq: equation22}, we obtain:
 \vspace{-0.1cm}
\begin{align}
\label{eq: fin1}
\mathbb{E} \| \boldsymbol{\cwb}_{t,i-1} \|^2 &\leq \left(1 - \sigma_{11} \mu + \mu^2 v_1^2 v_2^2 ({{\beta}}^{\text{multi}}_{s,\max,t})^2  \right) \mathbb{E} \| \boldsymbol{\cwb}_{t,i-1} \|^2 \notag \\
&+ \left( \frac{ \mu \sigma_{12}^2 }{ \sigma_{11} } + \mu^2 v_1^2 v_2^2  ({{\beta}}^{\text{multi}}_{s,\max,t})^2   \right) \mathbb{E} \| \bcwc_{t,i-1} \|^2 \notag\\
& + \mu^2 v_1^2 \expec[({\boldsymbol{{\sigma}}}^{\text{multi}}_{s,t})^2] ,
\end{align}
and
 \vspace{-0.2cm}
\begin{align}
\label{eq: fin2}
\mathbb{E} \| \bcwc_{t,i} \|^2 &\leq \left( \frac{3 \mu^2 \sigma_{21}^2}{1 - \rho(\mathcal{J}_{\epsilon}) - \epsilon} + \mu^2 v_1^2  v_2^2  ({{\beta}}^{\text{multi}}_{s,\max,t})^2   \right) \mathbb{E} \| \boldsymbol{\cwb}_{t,i-1} \|^2  \notag \\
&  + \left( \rho(\mathcal{J}_{\epsilon}) + \epsilon + \frac{3 \mu^2 \sigma_{22}^2}{1 - \rho(\mathcal{J}_{\epsilon}) - \epsilon} + \mu^2 v_1^2 v_2^2  ({{\beta}}^{\text{multi}}_{s,\max,t})^2  \right) \notag \\
& \quad  \mathbb{E} \| \bcwc_{t,i-1} \|^2 + \left( \frac{3\mu^2}{1 - \rho(\mathcal{J}_{\epsilon}) - \epsilon} \right)  \expec \left\| \boldsymbol{\widecheck {b}}_t^{\text{multi}} \right \|^2  \notag \\
& +  \mu^2 v_1^2 \expec[({\boldsymbol{{\sigma}}}^{\text{multi}}_{s,t})^2] .
 \vspace{-0.1cm}
\end{align}
By combining \eqref{eq: fin1} and \eqref{eq: fin2} into a single inequality recursion, we obtain~\eqref{eq: finfin}.

Then, by using the property that the spectral radius of a matrix is upper bounded by its $1$-norm, we obtain:
 \vspace{-0.1cm}
\begin{equation}
 \vspace{-0.1cm}
\rho(\Gamma) \leq \max \{ 1 - O(\mu) + O(\mu^2), \rho(\mathcal{J}_{\epsilon}) + \epsilon + O(\mu) + O(\mu^2) \}
 \vspace{-0.1cm}
\end{equation}
Since $\rho(\mathcal{J}_{\epsilon}) < 1$ is independent of $\mu$, and since $\epsilon$ and $\mu$ are small positive numbers that can be chosen arbitrarily small and independently of each other, it is clear that the RHS of the above expression can be made strictly smaller than one for sufficiently small $\epsilon$ and $\mu$. In that case $\rho(\Gamma) < 1$ so that $\Gamma$ is stable.  
By iterating~\eqref{eq: finfin}, we obtain~\eqref{eq: finaliteartion}, from which we conclude that $\limsup_{i \to \infty} \mathbb{E} \| \boldsymbol{\cwb}_{t,i} \|^2 = O(\mu)$ and $\limsup_{i \to \infty} \mathbb{E} \| \bcwc_{t,i} \|^2 = O(\mu^2)$. Therefore,
 \vspace{-0.2cm}
\begin{align}
\limsup_{i \to \infty} \mathbb{E} \| \bcwt_{t,i}  \|^2 &= \limsup_{i \to \infty} \mathbb{E} \left\| \mathcal{V}_{\epsilon}
\begin{bmatrix}
 \boldsymbol{\cwb}_{t,i}\\
\bcwc_{t,i} 
\end{bmatrix} \right\|^2 \notag \\
&\leq \limsup_{i \to \infty} v_2^2 \left[ \mathbb{E} \|  \boldsymbol{\cwb}_{t,i} \|^2 + \mathbb{E} \| \bcwc_{t,i}  \|^2 \right] \notag \\
& = O(\mu).
\end{align}
 \vspace{-0.1cm}

 \vspace{-0.8cm}
\section{Proof of Lemma~\ref{theoremcollab}}
\label{appendix: proof of theoremcollab}
We start by proving the non-cooperative bound~\eqref{eq: noncollab}. Under Assumption~\ref{ass: boundedhess}, the empirical multi-objective cost~$\bJ_{k,t}^{\text{multi}}$ defined in~\eqref{eq: jmulti} is $\nu$-strongly convex. Therefore,
\begin{equation}
\big\langle \nabla_w \bJ_{k,t}^{\text{multi}}(x)-\nabla_w \bJ_{k,t}^{\text{multi}}(y),\,x-y\big\rangle \ge  \nu \|x-y\|^2,\qquad \forall x,y. \label{eq:strong_mono}
\end{equation}
Using $\nabla_w \bJ_{k,t}^{\text{multi}}(\bw_{k,t}^\star)=0$ and setting $x=\bw_{k,t}^\star$ and $y=w_t^\star$, we obtain:
\begin{align}
\nu \|\bw_{k,t}^\star-w_t^\star \|^2 & \le \big\langle \nabla_w \bJ_{k,t}^{\text{multi}}(\bw_{k,t}^\star)-\nabla_w \bJ_{k,t}^{\text{multi}}(w_t^\star), \bw_{k,t}^\star-w_t^\star \big\rangle \notag\\
& \le \|\nabla_w \bJ_{k,t}^{\text{multi}}(w_t^\star)\| \|\bw_{k,t}^\star-w_t^\star\|. \label{eq: boundedd}
 \vspace{-0.1cm}
\end{align}
If $\bw_{k,t}^\star\neq w_t^\star$, we can divide both sides by $\|\bw_{k,t}^\star-w_t^\star \|$ to obtain
 \vspace{-0.1cm}
\begin{equation}
 \vspace{-0.1cm}
\|\bw_{k,t}^\star-w_t^\star\| \le \frac{1}{\nu} \|\nabla_w \bJ_{k,t}^{\text{multi}}(w_t^\star)\|.
\label{eq:perturb_basicc}
\end{equation}
Squaring~\eqref{eq:perturb_basicc} and taking expectations give:
\begin{equation}
\expec\|\bw_{k,t}^\star-w_t^\star \|^2 \le \frac{1}{\nu^2} \expec \|\nabla_w \bJ_{k,t}^{\text{multi}}(w_t^\star)\|^2.
\label{eq:perturb_sqq}
\end{equation}
Since $w_t^\star$ minimizes the deterministic multi-objective cost~$J_{k,t}^{\text{multi}}$, we have $\nabla_w J_{k,t}^{\text{multi}}(w_t^\star)=0$.
Therefore,
\begin{align}
\nabla_w \bJ_{k,t}^{\text{multi}}(w_t^\star) &= \nabla_w \bJ_{k,t}^{\text{multi}}(w_t^\star)-\nabla_w J^{\text{multi}}_{k,t}(w_t^\star) \notag\\
& = \frac{1}{t} \sum_{v=1}^{t-1}\lambda^{t-v} \boldsymbol{\Delta}_{k,v}(w_t^\star),
\label{eq:grad_diff_sum}
\end{align}
where
\begin{equation}
\boldsymbol{\Delta}_{k,v}(w) \triangleq \nabla_w \bJ^{\text{emp}}_{k,v}(w) -\nabla_w J_{k,v}(w). \label{eq: deltadef}
\end{equation}
 Using the definition of $\bJ_{k,v}^{\text{emp}}$ in~\eqref{eq: emp}, together with Assumption~\ref{ass: boundedhess}, it follows that $\boldsymbol{\Delta}_{k,v}(w)$ is zero-mean, and satisfies:
 \begin{equation}
\label{eq: finaldetlaa}
\expec \Vert \boldsymbol{\Delta}_{k,v}(w_t^\star) \Vert^2 \leq \frac{C_{k,v}}{Dp},
\end{equation}
where $C_{k,v} \triangleq 2\delta_{k,v} ^2\Vert w_t^\star - w^o_v\Vert^2 + 2 G_{k,v}^2 + (1-p) \Big\|\nabla_w J_{k,v}(w_t^\star)\Big\|^2$. A detailed derivation of this bound is provided in Appendix~H of the Supplementary Material.

To bound the norm in~\eqref{eq:perturb_sqq}, we apply Jensen inequality to write:
\begin{align}
 \expec\|\bw_{k,t}^\star-w_t^\star \|^2 & \le \frac{1}{\nu^2 } \expec \left\Vert   \frac{1}{t} \sum_{v=1}^{t-1}\lambda^{t-v} \boldsymbol{\Delta}_{k,v}(w_t^\star)\right\Vert ^2 \notag \\
& \le  \frac{t-1}{\nu^2 t^2 }   \sum_{v=1}^{t-1}\lambda^{2(t-v)} \expec \left\Vert  \boldsymbol{\Delta}_{k,v}(w_t^\star ) \right\Vert ^2.
\label{eq:wstar_wmulti_basic}
\vspace{-0.1cm}
\end{align}
Using the bound in~\eqref{eq: finaldetlaa},  we can write 
\begin{equation}
\expec\|\bw_{k,t}^\star-w_t^\star \|^2  \le \frac{t-1}{ \nu^2 t^2} \sum_{v=1}^{t-1}\lambda^{2(t-v)}\frac{\bar{C}}{Dp},
\end{equation}
where $\bar{C} \triangleq \max_{1\le k\le K}\max_{1\le v\le t-1}  C_{k,v}$. Since $0\le \lambda \le 1$, we have $\lambda^{2(t-v)}\le 1$ which implies that $ \sum_{v=1}^{t-1}\lambda^{2(t-v)} \le t-1$ and therefore~$\frac{(t-1)^2}{t^2} < 1$. Hence, we conclude~\eqref{eq: noncollab}.

We now consider the collaborative bound~\eqref{eq: benefit}. Under Assumption~\ref{ass: boundedhess}, the empirical global cost~$\bJ_{t}^{\text{glob}}$ defined in~\eqref{eq: global} is $K\nu$-strongly convex. Therefore,
\begin{equation}
\big\langle \nabla_w \bJ_{t}^{\text{glob}}(x)-\nabla_w \bJ_{t}^{\text{glob}}(y),\,x-y\big\rangle \ge K\nu \|x-y\|^2,\qquad \forall x,y. \label{eq:strong_mono}
\end{equation}
Using $\nabla_w \bJ_{t}^{\text{glob}}(\bw_t^\star)=0$, setting $x=\bw_t^\star$ and $y=w_t^\star$, and using arguments similar to those used in~\eqref{eq: boundedd},~\eqref{eq:perturb_basicc} and~\eqref{eq:perturb_sqq},  we can show that:
\vspace{-0.2cm}
\begin{equation}
\vspace{-0.2cm}
\expec\|\bw_t^\star-w_t^\star \|^2 \le \frac{1}{K^2\nu^2} \expec \|\nabla_w \bJ_{t}^{\text{glob}}(w_t^\star)\|^2.
\label{eq:perturb_sq}
\end{equation}
Since $w_t^\star$ minimizes the global cost~$J_t^{\text{glob}}$, we have $\nabla_w J_t^{\text{glob}}(w_t^\star)=0$.
Therefore,
\vspace{-0.1cm}
\begin{align}
\nabla_w \bJ_t^{\text{glob}}(w_t^\star) &= \nabla_w \bJ_t^{\text{glob}}(w_t^\star)-\nabla_w J^{\text{glob}}_t(w_t^\star) \notag\\
& = \sum_{k=1}^K \frac{1}{t} \sum_{v=1}^{t-1}\lambda^{t-v} \boldsymbol{\Delta}_{k,v}(w_t^\star),
\label{eq:grad_diff_sum}
\end{align}
where~$\boldsymbol{\Delta}_{k,v}(w)$ is defined in~\eqref{eq: deltadef}.

Therefore to bound~\eqref{eq:perturb_sq}, we use~\eqref{eq:grad_diff_sum} and apply Jensen's inequality to write
\begin{equation}
 \expec\|\bw_t^\star-w_t^\star \|^2  \le  \frac{t-1}{K^2\nu^2 t^2 } \sum_{k=1}^K   \sum_{v=1}^{t-1}\lambda^{2(t-v)} \expec \left\Vert  \boldsymbol{\Delta}_{k,v}(w_t^\star ) \right\Vert ^2 .
\label{eq:wstar_wmulti_basic}
\end{equation}
Using the bound in~\eqref{eq: finaldetlaa},  we can write:
\begin{equation}
\expec\|\bw_t^\star-w_t^\star \|^2  \le \frac{t-1}{K \nu^2 t^2} \sum_{v=1}^{t-1}\lambda^{2(t-v)}\frac{\bar{C}}{Dp},
\end{equation} 
from which we conclude that~\eqref{eq: benefit} holds.

\section*{Supplementary Material}

\addcontentsline{toc}{section}{Supplementary Material}

\renewcommand{\thesection}{\arabic{section}}
\section{Derivation of inequality~\eqref{eq: equation1}}
\label{ bound of d11}

By applying Jensen's inequality~\cite{jensen1906fonctions} to the convex function $\|x\|^2$, we can bound the first term on the RHS of~\eqref{eq: jensenfirstequation} as follows:
\begin{align}
    &\mathbb{E}\left\| (I_{M} - \mu \bcD_{11,t,i-1}) \boldsymbol{\cwb}_{t,i-1}  - \mu \bcD_{12,t,i-1}  \bcwc_{t,i-1} \right\|^2 \notag \\
    &\leq \frac{1}{1 - z} \mathbb{E}\left[ \left\| (I_{M} -  \mu \bcD_{11,t,i-1}) \right\|^2 \left\| \boldsymbol{\cwb}_{t,i-1} \right\|^2\right]   \notag \\
    & \qquad  + \frac{1}{z} \mathbb{E} \left[ \mu^2\left\| \bcD_{12,t,i-1}\right\|^2  \left\|\bcwc_{t,i-1} \right\|^2\right],
    \label{eq: jensenD11}
\end{align}
for any arbitrary positive number $z \in (0,1)$. 
Using the linearity of the gradient and Hessian operators, we can express the matrix ~$\bH_{k,t,i-1}^{\text{multi}}$ defined in~\eqref{eq: hermitian} as:
\begin{equation}\label{eq: Hfin}
 \bH_{k,t,i-1}^{\text{multi}}= \frac{1}{t}\left(\bH_{k,t,i-1} + \sum_{v=1}^{t-1} \bH^{\text{emp}}_{k,v,i-1}\right),
\end{equation}
where $\bH_{k,t,i-1} \triangleq \int_{0}^{1} \nabla_{w} ^2 J_{k,t}( \boldsymbol{w}^\star_{t} + \tau \bwt_{k,t,i-1})d\tau$ , and $\bH_{k,v,i-1}^{\text{emp}} \triangleq \int_{0}^{1} \nabla_{w} ^2 \bJ_{k,v}^{\text{emp}}( \boldsymbol{w}^\star_{t} + \tau \bwt_{k,t,i-1})d\tau$.
From Assumption~\ref{ass: boundedhess} and Lemma~\ref{lemma1}, $\bH_{k,t,i-1}^{\text{multi}}$ can be bounded as follows:
\begin{equation}   \label{eq: hermitian22}
  \bH_{k,t,i-1}^{\text{multi}} \geq \frac{1}{t} \nu_{k,t}  I_{M}  \geq   \left( \min_{\substack{1\leq t \leq T\\[1mm] 1\leq k \leq K}} \frac{\nu_{k,t}}{t} \right) I_{M} \triangleq \nu I_M > 0,
\end{equation}
where $\nu \triangleq \left( \min_{\substack{1\leq t \leq T\\[1mm] 1\leq k \leq K}} \frac{\nu_{k,t}}{t} \right)$, and
\begin{align}  \label{eq: hermitian2}
 \bH_{k,t,i-1}^{\text{multi}} &  \leq \frac{1}{t} \left(\sum_{v=1}^{t-1} \lambda^{t-v}\boldsymbol{\delta}_{k,v}^{\text{emp}}+\delta_{k,t} \right) I_{M} \notag\\
&  \leq \frac{1}{t} \left(\sum_{v=1}^{t-1} \lambda^{t-v}  \frac{\delta_{k,v}}{p_{k,v}}  +\delta_{k,t} \right) I_{M} \notag \\
& \leq \frac{1}{t} \left(\sum_{v=1}^{t-1} \lambda^{t-v}  \frac{  \max_{1\leq v \leq t-1}  \delta_{k,v} }{\min_{1\leq v \leq t-1}p_{k,v}} + \delta_{k,t} \right) I_{M} \notag \\
&  \leq   \max_{1\leq t \leq T} \frac{1}{t}\left (\frac{\lambda}{1-\lambda} \frac{ \max_{1\leq v \leq t-1}  \delta_{k,v} }{ \min_{1\leq v \leq t-1}p_{k,v}} + \delta_{k,t}\right) I_{M}. 
\end{align}
Using now \eqref{eq: hermitian22}, and \eqref{eq: hermitian2}, $\bcD_{11,t,i-1}$ can be bounded as follows:
\begin{equation}
 0 <  \nu I_{M} \leq  \frac{1}{{K}} \sum_{k=1}^K \bH_{k,t,i-1}^{\text{multi}} \leq  \delta I_{M},
 \end{equation}
 where $\delta \triangleq \max_{\substack{1\leq t \leq T\\[1mm] 1\leq k \leq K}}   \frac{1}{t}\left (\frac{\lambda}{1-\lambda} \frac{ \max_{1\leq v \leq t-1}  \delta_{k,v} }{ \min_{1\leq v \leq t-1}p_{k,v}} + \delta_{k,t}\right)$
that is independent of $\mu$ and $i$. In terms of the 2-induced matrix norm (i.e., maximum singular value), we can characterize~$ \| I_{M} - \mu \bcD_{11,t,i-1} \| $ as follows :
\begin{align}
    \| I_{M} - \mu \bcD_{11,t,i-1} \| &= \rho (I_{M} - \mu \bcD_{11,t,i-1}) \notag\\
   & \leq \max \left\{ \vert 1 - \mu {\delta} \vert , \vert 1 - \mu \nu \vert \right\},
\end{align}
and, therefore,
\begin{equation}
\label{eq: first}
\| I_{M} - \mu \bcD_{11,t,i-1} \|^2 \leq (1 - \mu \sigma_{11})^2,
\end{equation}
for some positive constant~$\sigma_{11}$ that is independent of~$\mu$ and~$i$.
Similarly, using the 2-induced matrix norm (i.e., maximum singular value), we can bound $\| \bcD_{12,t,i-1} \|^2$ as follows:
\begin{align}
\| \bcD_{12,t,i-1} \|^2 &=  \left\| \frac{1}{\sqrt{K}}(\mathds{1}_K^{\top} \otimes I_M) \bcH_{t,i-1}^{\text{multi}} \mathcal{V}_{R,\epsilon} \right\|^2  \notag \\
&\leq  \left( \max_{1 \leq k \leq K} \left\| \bH_{k,t,i-1}^{\text{multi}} \right\|^2 \right) \left\|  \mathcal{V}_{R,\epsilon} \right\|^2  \notag\\
&\leq  \| \mathcal{V}_{R,\epsilon}  \|^2 \delta^2 =  \sigma_{12}^2, \label{eq: second}
\end{align}
for some positive constant $\sigma_{12}$.
Substituting~\eqref{eq: jensenD11} into~\eqref{eq: jensenfirstequation}, using~\eqref{eq: first}, and ~\eqref{eq: second}, and selecting~$z=\mu \sigma_{11}$, we obtain~\eqref{eq: equation1} for sufficiently small~$\mu$.

\section{Derivation of inequality~\eqref{eq: equation22}}
\label{ bound the secondrelation}
We repeat similar arguments for the second variance relation~\eqref{eq: jensensecond}. Using Jensen's inequality again, we obtain:
\begin{align}
& \mathbb{E} \left\| (\mathcal{J}_\epsilon - \mu \bcD_{22,t,i-1}) \bcwc_{t,i-1}  - \mu \bcD_{21,t,i-1} \boldsymbol{\cwb}_{t,i-1}  + \mu \boldsymbol{\widecheck {b}}_t^{\text{multi}} \right\|^2 \notag\\
&\leq \frac{1}{z} \mathbb{E} \left\| \mathcal{J}_\epsilon \bcwc_{t,i-1} \right\|^2 + \frac{1}{1 - z} \mathbb{E} \Big\| \mu \bcD_{22,t,i-1} \bcwc_{t,i-1}  \notag\\
&  \qquad + \mu \bcD_{21,t,i-1} \boldsymbol{\cwb}_{t,i-1}  - \mu \boldsymbol{\widecheck {b}}_t^{\text{multi}}\Big \|^2 \notag\\
&\overset{(a)}{\leq} \frac{(\rho(\mathcal{J}_\epsilon) + \epsilon)^2}{z} \mathbb{E} \| \bcwc_{t,i-1} \|^2 + \frac{1}{1 - z} \mathbb{E}\Big \| \mu \bcD_{22,t,i-1} \bcwc_{t,i-1}  \notag\\
& \qquad + \mu \bcD_{21,t,i-1} \boldsymbol{\cwb}_{t,i-1}  - \mu \boldsymbol{\widecheck {b}}_t^{\text{multi}} \Big\|^2 \notag\\
&\overset{(b)}{=} (\rho(\mathcal{J}_\epsilon) + \epsilon) \mathbb{E} \| \bcwc_{t,i-1} \|^2 \notag\\
& +   \frac{\mu^2}{1 - \rho(\mathcal{J}_\epsilon) - \epsilon} \mathbb{E} \Big \|  \bcD_{22,t,i-1} \bcwc_{t,i-1} +  \bcD_{21,t,i-1} \boldsymbol{\cwb}_{t,i-1} \notag\\
& \qquad   - \boldsymbol{\widecheck {b}}_t^{\text{multi}}\Big \|^2, \label{eq: jenseneq3}
\end{align}
for any arbitrary positive number $z \in (0,1)$. In~$(a)$ we used the fact that the block diagonal matrix~$\mathcal{J}_\epsilon$ satisfies~\cite[ p. 516]{sayed2014adaptation}:
\begin{equation}
\label{eq: jeps}
\| \mathcal{J}_\epsilon \|^2 \leq (\rho(\mathcal{J}_\epsilon) + \epsilon)^2. 
\end{equation}
In $(b)$, we used the fact that~$\rho(\mathcal{J}_{\epsilon}) \in (0,1)$, and thus,~$\epsilon$ can be selected small enough to ensure~$\rho(\mathcal{J}_{\epsilon}) + \epsilon \in (0,1)$. Moreover, we selected~$z = \rho(\mathcal{J}_{\epsilon}) + \epsilon$. Using Jensen’s inequality, the second term on the RHS of~\eqref{eq: jenseneq3} can be bounded as follows:
\begin{align}
&\mathbb{E} \left\| \bcD_{22,t,i-1} \bcwc_{t,i-1} + \bcD_{21,t,i-1} \boldsymbol{\cwb}_{t,i-1}  - \boldsymbol{\widecheck {b}}_t^{\text{multi}} \right\|^2 \notag\\
&\leq 3 \mathbb{E} \left[  \| \bcD_{22,t,i-1} \|^2 \| \bcwc_{t,i-1} \|^2\right] +  3  \expec \left\|  \boldsymbol{\widecheck {b}}_t^{\text{multi}} \right\|^2 \notag\\
&\qquad  +  3 \mathbb{E} \left[  \left\| \bcD_{21,t,i-1} \|^2 \| \boldsymbol{\cwb}_{t,i-1}\right\|^2\right] . \label{eq: jensen333}
\end{align}
Following similar arguments as in~\eqref{eq: second}, we can show that:
\begin{equation}
\label{eq: d22d21}
\| \bcD_{21,t,i-1} \|^2 \leq  \sigma_{21}^2, \quad \| \bcD_{22,t,i-1} \|^2 \leq  \sigma_{22}^2,
\end{equation}
for some positive constants~$\sigma_{21}$ and $\sigma_{22}$. Substituting~\eqref{eq: jensen333} into~\eqref{eq: jenseneq3} and~\eqref{eq: jenseneq3} into~\eqref{eq: jensensecond}, and using~\eqref{eq: d22d21}, we obtain~\eqref{eq: equation22}.

\section{Proof of Lemma~\ref{lemma2}}
\label{appendix: proof of lemma 2}
To establish~\eqref{eq: finitewstar}, we start by noting that under Assumption~\ref{ass: boundedhess}, the global aggregate cost at task~$t$, $\bJ_t^{\text{glob}}(w)$, defined in~\eqref{eq: global} is  $K \nu$-strongly convex, and has a unique solution~$\bw^\star_t$, defined in~\eqref{eq: wstar} that satisfies:
\begin{equation}
\label{eq: jglobopt}
  \nabla_w \bJ_t^{\text{glob}}(\bw_t^\star) = 0.
\end{equation}
It follows that, for any $w_1,w_2 \in \mathbb{R}^M$,
\begin{equation}
  \big(\nabla_w \bJ_t^{\text{glob}}(w_1) - \nabla_w \bJ_t^{\text{glob}}(w_2)\big)^\top (w_1-w_2) \ge  K\nu  \|w_1-w_2\|^2.
\end{equation}
By choosing $w_1= \bw_t^\star$ and $w_2=w^o_u$, and using~\eqref{eq: jglobopt}, we obtain:
\begin{equation}
  - \nabla_w \bJ_t^{\text{glob}}(w^o_u)^\top (\bw_t^\star - w^o_u)  \ge  {K\nu}   \|\bw_t^\star - w^o_u\|^2.
  \label{eq:wstar-strong-convex-final}
\end{equation}
Applying the Cauchy-Schwarz inequality to the LHS yields:
\begin{align}
   K\nu \|\bw_t^\star - w^o_u\|^2 & \le \big|\nabla_w \bJ^{\text{glob}}_t(w^o_u)^\top (\bw_t^\star - w^o_u)\big| \notag\\
 & \le  \big \|\nabla_w \bJ^{\text{glob}}_t(w^o_u)\big\|  \|\bw_t^\star - w^o_u\|.
\end{align}
If $\bw_t^\star \neq w^o_u$, we can divide both sides by $\|\bw_t^\star - w^o_u\|$ to obtain:
\begin{equation}
  \|\bw_t^\star - w^o_u\|  \le  \frac{1}{K\nu} \big\|\nabla_w \bJ^{\text{glob}}_t(w^o_u)\big\|.
\end{equation}
When~$\bw_t^\star=w^o_u$, the same bound holds trivially.
 Squaring both sides and taking expectations, we get:
\begin{equation}
  \mathbb{E} \|\bw_t^\star- w^o_u\|^2  \le  \frac{1}{ K^2\nu ^2} \mathbb{E}  \|\nabla_w \bJ^{\text{glob}}_t(w^o_u)\|^2.
\end{equation}
It remains to show that the right-hand side is finite. Using the fact that:
\begin{equation}
 \nabla_w \bJ^{\text{glob}}_t(w^o_u) \triangleq \sum_{k=1}^K \nabla_w \bJ^{\text{multi}}_{k,t}(w^o_u).
\end{equation}
and applying Jensen's inequality, we obtain:
\begin{align}
  \mathbb{E} \|\nabla_w \bJ^{\text{glob}}_t(w^o_u)\|^2 &=  \mathbb{E}   \left\|  \sum_{k=1}^K \nabla_w \bJ^{\text{multi}}_{k,t}(w^o_u)  \|^2  \right\| \notag\\
  &\le K\sum_{k=1}^K \mathbb{E} \|\nabla_w \bJ^{\text{multi}}_{k,t}(w^o_u)\|^2 . \label{eq: jjglob}
\end{align}
Since~$ \bJ^{\text{multi}}_{k,t}$ is a weighted sum of the current expected risk~$J_{k,t}$ and past empirical risks~$\bJ_{k,v}^{\text{emp}}$, it is enough to show that the second moments of their gradients at~$w^o_v$ are finite. For the empirical risk term, we have:

\begin{align}
  & \expec \left \Vert \nabla_w \bJ^{\text{emp}}_{k,v}(w^o_u) \right \Vert^2 \notag\\
  & =  \expec \left\|\frac{1}{p_{k,v} D_v} \sum_{n=1}^{D_v} \boldsymbol{\alpha}_{k,v}(n) \nabla_w Q_{k,v}(w^o_u; \bx_{k,v,n})\right\|^2 .
  \end{align}
  Using Jensen's inequality, $\Vert \sum_{n=1}^N x_n \Vert^2 \leq N \sum_{n=1}^N \Vert x_n\Vert^2$, we obtain:
  \begin{align}
   & \expec \left \Vert \nabla_w \bJ^{\text{emp}}_{k,v}(w^o_u) \right \Vert^2 \notag\\
  &\le \frac{D_v}{(p_{k,v} D_v)^2} \sum_{n=1}^{D_v} \expec \left[ (\boldsymbol{\alpha}_{k,v}(n))^2 \big \|\nabla_w Q_{k,v}(w^o_u; \bx_{k,v,n})\big\|^2\right] \notag\\
  &\leq  \frac{1}{p_{k,v} D_v} \sum_{n=1}^{D_v}  \expec \big\|\nabla_w Q_{k,v}(w^o_u; \bx_{k,v,n})\big\|^2.
  \end{align}
 We now bound the sample-loss gradient at~$w_u^o$. To that end, we add and subtract the vector  $\nabla_{w} Q_{k,v}(w^o_v;\bx_{k,v,n})$, then apply Jensen's inequality~\cite{jensen1906fonctions} together with Assumption~\ref{ass: boundedhess} to obtain:
\begin{align}
\|\nabla_w Q_{k,v}(w_u^o; \bx_{k,v,n})\|^2&  \le 2\delta_{k,v}^2\|w_u^o-w_v^o\|^2 \notag\\
& \quad + 2\|\nabla_w Q_{k,v}(w_v^o; \bx_{k,v,n})\|^2 .
\end{align}
Taking expectations and using Assumption~\ref{ass:finite-moment-Q}, we can write:
\begin{equation}
\mathbb{E}
\|\nabla_w Q_{k,v}(w_u^o; \bx_{k,v,n})\|^2 \le 2\delta_{k,v}^2\|w_u^o-w_v^o\|^2 + 2G_{k,v}^2 <\infty .
\end{equation}
Therefore,
\begin{equation}
\expec \left \Vert \nabla_w \bJ^{\text{emp}}_{k,v}(w^o_u) \right \Vert^2 <\infty .
\end{equation}
Similarly, since the current risk~$J_{k,t}$ is twice differentiable and~$w_u^o$ is fixed, we have~$\|\nabla_w J_{k,t}(w_u^o)\|^2<\infty$. Consequently,
\begin{equation}
\mathbb{E}
\|\nabla_w \bJ_{k,t}^{\text{multi}}(w_u^o)\|^2<\infty, \label{eq: 137}
\end{equation}
and therefore~$\mathbb{E} \|\nabla_w \bJ_t^{\text{glob}}(w_u^o)\|^2<\infty$. It follows that~\eqref{eq: finitewstar} holds, with the finite constant~$C_u$ given by:
\begin{align}
C_u & \triangleq \frac{t}{K \nu^2} \sum_{k=1}^K \Big[ \|\nabla_w J_{k,t}(w_u^o)\|^2 \notag\\
& \quad +  \sum_{v=1}^{t-1}\frac{ \lambda^{2(t-v)} }{p_{k,v}} \left(2\delta_{k,v}^2\|w_u^o-w_v^o\|^2 + 2G_{k,v}^2\right) \Big].
\end{align}

We now prove~\eqref{eq: finiteb}. From the definition of $\bJ^{\text{multi}}_{k,t}$, together with Assumption~\ref{ass: boundedhess} and Lemma~\ref{lemma1}, it follows that the gradient $ \nabla_w \bJ^{\text{multi}}_{k,t}$ is $\delta_k$-Lipschitz continuous, where~$\delta_k \triangleq \max_{1\leq t \leq T} \frac{1}{t}\left (\frac{\lambda}{1-\lambda} \frac{ \max_{1\leq v \leq t-1}  \delta_{k,v} }{ \min_{1\leq v \leq t-1}p_{k,v}} + \delta_{k,t}\right)$. Hence, for any $w_1,w_2 \in \mathbb{R}^M$:
\begin{equation}
  \big\|\nabla_w \bJ^{\text{multi}}_{k,t}(w_1) - \nabla_w \bJ^{\text{multi}}_{k,t}(w_2)\big\|  \le  \delta_k  \|w_1 - w_2\|.
\end{equation}
Let~$u\in\{1,\ldots,t\}$. By choosing~$w_1=\bw_t^\star$ and $w_2=w_u^o$, we obtain
\begin{align}
  \| \boldsymbol{b}^{\text{multi}}_{k,t}\| &=  \big\|\nabla_w \bJ^{\text{multi}}_{k,t}(\bw_t^\star)\big\| \notag\\
  &\le \big\|\nabla_w \bJ^{\text{multi}}_{k,t}(\bw_t^\star) - \nabla_w \bJ^{\text{multi}}_{k,t}(w^o_u)\big\|  + \big\|\nabla_w \bJ^{\text{multi}}_{k,t}(w^o_u)\big\| \notag\\
  &\le  \delta_k \|\bw_t^\star - w^o_u\|  + \big\|\nabla_w \bJ^{\text{multi}}_{k,t}(w^o_u)\big\|.
\end{align}
Squaring both sides and applying the inequality $(a+b)^2 \le 2a^2 + 2b^2$ yields:
\begin{equation}
  \| \boldsymbol{b}^{\text{multi}}_{k,t}\|^2  \le  2 \delta_k^2 \|\bw_t^\star - w^o_u\|^2 + 2 \big\|\nabla_w \bJ^{\text{multi}}_{k,t}(w^o_u)\big\|^2.
\end{equation}
Taking expectations and using~\eqref{eq: finitewstar} and~\eqref{eq: 137}, we conclude that there exists a finite constant~$C_{2,k}$ such that:
\begin{equation}
\mathbb{E}
\|\boldsymbol b^{\text{multi}}_{k,t}\|^2 \le C_{2,k}<\infty.
\end{equation}
Property~\eqref{eq: finiteb} follows from the fact that:
\begin{equation}
\mathbb{E} \|\boldsymbol b^{\text{multi}}_t\|^2 = \sum_{k=1}^K \mathbb{E} \|\boldsymbol b^{\text{multi}}_{k,t}\|^2 \le \sum_{k=1}^K C_{2,k} \triangleq C_2<\infty.
\end{equation}

\section{Derivation of inequality~\eqref{eq: 87}}
\label{bound the noise}
From Assumption~\ref{ass: gradient noise}, Lemma~\ref{lemma1}, and Jensen's inequality, we have:
\begin{align}
&\mathbb{E} \left[\left\|  \bs_{k,t,i}^{\text{multi}} \right\|^2 \mid  \bcF_{t,i-1}  \right] \notag \\
&= \frac{1}{t^2}\mathbb{E} \left[ \left\|   \sum_{v=1}^{t-1} \lambda^{t-v} {\bs}_{k,v,i}^{\text{emp}}( \bw_{k,t,i-1}) +  \bs_{k,t,i}(\bw_{k,t,i-1}) \right\|^2  \mid \bcF_{t,i-1}  \right] \notag \\
& \leq \frac{2}{t^2}\mathbb{E} \left[ \left\|   \sum_{v=1}^{t-1} \lambda^{t-v} {\bs}_{k,v,i}^{\text{emp}}( \bw_{k,t,i-1})\right\|^2  \mid \bcF_{t,i-1}  \right]   \notag \\ 
& + \frac{2}{t^2} \mathbb{E} \left[ \left\|  \bs_{k,t,i}(\bw_{k,t,i-1}) \right\|^2  \mid \bcF_{t,i-1}  \right] \notag \\ 
& \leq \frac{2(t-1)}{t^2} \sum_{v=1}^{t-1} \lambda^{2(t-v)} \mathbb{E} \left[ \left\|   {\bs}_{k,v,i}^{\text{emp}}( \bw_{k,t,i-1}) \right\|^2 \mid  \bcF_{t,i-1} \right]   \notag\\
& + \frac{2}{t^2} \mathbb{E} \left[ \left\|  \bs_{k,t,i}(\bw_{k,t,i-1}) \right\|^2  \mid \bcF_{t,i-1}  \right] \notag \\ 
& \leq \frac{2(t-1)}{t^2} \sum_{v=1}^{t-1} \lambda^{2(t-v)}\left( (\boldsymbol{\beta}^{\text{emp}}_{s,k,v})^2 \| w^o_v - \bw_{k,t,i-1} \|^2 + (\boldsymbol{\sigma}^{\text{emp}}_{s,k,v})^2\right)  \notag\\
& + \frac{2}{t^2} \left(\beta_{s,k,t}^2 \| w^o_t - \bw_{k,t,i-1} \|^2 + \sigma_{s,k,t}^2 \right)\notag \\ 
& \leq \frac{2(t-1)}{t^2} \sum_{v=1}^{t-1} \lambda^{2(t-v)}\Big(4 \left(\frac{\delta_{k,v}}{p_{k,v}}\right)^2( \| \tilde{\bw}_{k,t,i-1} \|^2 \notag\\
& \quad +  \Vert w^o_v -  \boldsymbol{w}_{t}^\star \Vert^2) + (\boldsymbol{\sigma}^{\text{emp}}_{s,k,v})^2 \Big)  \notag\\
&  \quad + \frac{2}{t^2} \left(2 \beta_{s,k,t}^2( \| \tilde{\bw}_{k,t,i-1} \|^2 + \Vert w^o_t - \boldsymbol{w}_{t}^\star \Vert^2) + \sigma_{s,k,t}^2 \right)\notag \\
& = (\beta^{\text{multi}}_{s,k,t})^2 \| \tilde{\bw}_{k,t,i-1} \|^2 +  (\boldsymbol{\sigma}^{\text{multi}}_{s,k,t})^2 ,
\end{align}
where
\begin{equation}
({\beta}^{\text{multi}}_{s,k,t})^2 = \frac{8(t-1)}{t^2} \sum_{v=1}^{t-1} \lambda^{2(t-v)}   \left(\frac{\delta_{k,v}}{p_{k,v}}\right) ^2 + \frac{4}{t^2}  \beta_{s,k,t}^2,
\end{equation}
and $(\boldsymbol{{\sigma}}^{\text{multi}}_{s,k,t})^2 $ is defined in~\eqref{eq: sigmamulti}.

\section{Derivation of inequality~\eqref{eq: finaldetlaa}}
\label{appendix: proof of the bound}

Note from the definition of $\bJ_{k,v}^{\text{emp}}$ in~\eqref{eq: emp}, and the fact that $\expec[\boldsymbol{\alpha}_{k,v}(i) \mid \bx_{k,v,i}]=p_{k,v}$, that the vector~$\boldsymbol{\Delta}_{k,v}(w)$ is zero-mean, i.e., $\expec[\boldsymbol{\Delta}_{k,v}(w)]=0$. 
We next bound its second-order moment. We start by writing:
\begin{equation}
\boldsymbol{\Delta}_{k,v}(w) = \frac{1}{D_v}\sum_{i=1}^{D_v}  \left(\frac{\boldsymbol{\alpha}_{k,v}(i)}{p_{k,v}} \nabla_w Q_{k,v}(w;\bx_{k,v,i}) - \nabla_w J_{k,v}(w)\right)
\end{equation}
expanding the squared norm and taking expectations yields
\begin{align}
& \mathbb{E}\Big\| \boldsymbol{\Delta}_{k,v}(w)  \Big\|^2 \notag\\
&= \mathbb{E}\left\|  \frac{1}{D_v}\sum_{i=1}^{D_v}  \left(\frac{\boldsymbol{\alpha}_{k,v}(i)}{p_{k,v}} \nabla_w Q_{k,v}(w;\bx_{k,v,i}) - \nabla_w J_{k,v}(w)\right) \right\|^2.
\end{align}
For each task~$v$ and agent~$k$, the pairs $\{(\boldsymbol{\alpha}_{k,v}(i),\bx_{k,v,i})\}_{i=1}^{D_v}$ are independent across sample indices~$i$ and agent indices~$k$.
Moreover, for every~$i$, the storage indicator $\boldsymbol{\alpha}_{k,v}(i)$ is independent of the corresponding sample $\bx_{k,v,i}$.
It then follows that the cross terms vanish and only diagonal terms remain:
\begin{align}
& \mathbb{E}\Big\|  \boldsymbol{\Delta}_{k,v}(w)  \Big\|^2 \notag\\
& = \frac{1}{D_v^2}\sum_{i=1}^{D_v}\mathbb{E} \left\Vert \frac{\boldsymbol{\alpha}_{k,v}(i)}{p_{k,v}} \nabla_w Q_{k,v}(w;\bx_{k,v,i}) - \nabla_w J_{k,v}(w) \right\Vert^2 \notag\\
&  = \frac{1}{D_v} \mathbb{E} \left\Vert \frac{\boldsymbol{\alpha}_{k,v}(i)}{p_{k,v}} \nabla_w Q_{k,v}(w;\bx_{k,v,i}) - \nabla_w J_{k,v}(w) \right\Vert^2. \label{eq: deltabold}
\end{align}
Since $\mathbb{E}\left[\frac{\boldsymbol{\alpha}_{k,v}(i)}{p_{k,v}} \nabla_w Q_{k,v}(w;\bx_{k,v,i})\right]= \nabla_w J_{k,v}(w) $, we have:
\begin{align}
& \mathbb{E} \left\|\frac{\boldsymbol{\alpha}_{k,v}(i)}{p_{k,v}} \nabla_w Q_{k,v}(w;\bx_{k,v,i}) - \nabla_w J_{k,v}(w) \right\|^2 \notag\\
& =\mathbb{E}\left\|\frac{\boldsymbol{\alpha}_{k,v}(i)}{p_{k,v}} \nabla_w Q_{k,v}(w;\bx_{k,v,i}) \right \|^2- \left\| \nabla_w J_{k,v}(w) \right\|^2.\notag\\
& = \frac{1}{p_{k,v}}  \mathbb{E}\| \nabla_w Q_{k,v}(w;\bx_{k,v,i}) \|^2 - \left\| \nabla_w J_{k,v}(w) \right\|^2 \notag\\
& = \frac{1}{p_{k,v}}\mathbb{E} \left\| \nabla_w Q_{k,v}(w;\bx_{k,v,i})  - \nabla_w J_{k,v}(w) \right\|^2 \notag\\
& \quad \quad  + \frac{1-p_{k,v}}{p_{k,v}} \left\| \nabla_w J_{k,v}(w) \right\|^2 , \label{eq: beforei}
\end{align}

where the last equality follows from $\mathbb{E}\| \nabla_w Q_{k,v}(w;\bx_{k,v,i})  \|^2=\mathbb{E}\| \nabla_w Q_{k,v}(w;\bx_{k,v,i})  - \nabla_w J_{k,v}(w) \|^2+\| \nabla_w J_{k,v}(w)\|^2$.
Combining~\eqref{eq: beforei} with~\eqref{eq: deltabold}, we obtain:
\begin{align}
&\mathbb{E}\Big\|  \boldsymbol{\Delta}_{k,v}(w) \Big\|^2 =\frac{1}{D_v p_{k,v}} \mathbb{E}\Big\| \nabla_w Q_{k,v}(w;\bx_{k,v,i})-\nabla_w J_{k,v}(w)\Big\|^2 \notag\\
& \qquad \qquad \qquad + \frac{1-p_{k,v}}{D_v p_{k,v}} \Big\|\nabla_w J_{k,v}(w)\Big\|^2. \label{eq: withw}
\end{align}
Evaluating  \eqref{eq: withw} at $w= w_t^\star$,  we get:
\begin{align}
&\mathbb{E}\Big\|  \boldsymbol{\Delta}_{k,v}(w_t^\star)  \Big\|^2 \notag\\
& =\frac{1}{D_v p_{k,v}} \mathbb{E}\Big\| \nabla_w Q_{k,v}(w_t^\star;\bx_{k,v,i})-\nabla_w J_{k,v}(w_t^\star)\Big\|^2 \notag\\
& \qquad + \frac{1-p_{k,v}}{D_v p_{k,v}} \Big\|\nabla_w J_{k,v}(w_t^\star)\Big\|^2 \notag\\
& \leq  \frac{1}{D_v p_{k,v}} \mathbb{E}\Big\|\nabla_w Q_{k,v}(w_t^\star; \bx_{k,v,i})\Big\|^2 + \frac{1-p_{k,v}}{D_v p_{k,v}} \Big\|\nabla_w J_{k,v}(w_t^\star)\Big\|^2. \label{eq: deltabold2}
\end{align}
By Assumption~\ref{ass: boundedhess}, we have that
\begin{align}
&\Big\| \nabla_w Q_{k,v}(w_t^\star ; \bx_{k,v,i}) \Big\|^2 \notag\\
&\leq 2\delta_{k,v}^2 \|w_t^\star-w_v^o\|^2 + 2 \Big\| \nabla_w Q_{k,v}(w_v^o ; \bx_{k,v,i}) \Big\|^2 .
\end{align}
Taking expectation and using Assumption~\ref{ass:finite-moment-Q}, we obtain:
\begin{equation}
\expec \Big\| \nabla_w Q_{k,v}(w_t^\star ; \bx_{k,v,i}) \Big\|^2 \leq 2\delta_{k,v}^2  \|w_t^\star-w_v^o\|^2 + 2G_{k,v}^2 . \label{eq: dd}
\end{equation}
Therefore, combining~\eqref{eq: deltabold2} with~\eqref{eq: dd}, we get:
\begin{align}
&\expec \Big\| \boldsymbol{\Delta}_{k,v}(w_t^\star) \Big\|^2 \notag\\
 & \leq  \frac{1}{D_v p_{k,v}}   \left( 2\delta_{k,v} ^2\Vert w_t^\star - w^o_v\Vert^2+ 2G_{k,v}^2 \right) \notag\\
 & \qquad + \frac{1-p_{k,v}}{D_v p_{k,v}} \Big\|\nabla_w J_{k,v}(w_t^\star )\Big\|^2.
\end{align}
Using the fact that $J_{k,t}^{\text{multi}}$ is strongly convex, together with the fact that~$p=p_{k,v}$ for all $(k,v)$, and $D=D_v$ for all $v$, we can finally conclude that :
\begin{equation}
\label{eq: finaldetla}
\expec \Vert \boldsymbol{\Delta}_{k,v}(w_t^\star) \Vert^2 \leq \frac{C_{k,v}}{Dp},
\end{equation}
where $C_{k,v} \triangleq 2\delta_{k,v} ^2\Vert w_t^\star - w^o_v\Vert^2 + 2 G_{k,v}^2 + (1-p) \Big\|\nabla_w J_{k,v}(w_t^\star)\Big\|^2$.

{\balance{\bibliographystyle{./IEEEtran}
\bibliography{./reference}}}

\end{document}